\RequirePackage{fix-cm}
\documentclass[12pt]{article}
\usepackage[latin9]{inputenc}
\usepackage{geometry}
\usepackage{verbatim} 
\usepackage{float}
\usepackage{booktabs}
\usepackage{amsthm}
\usepackage{amsmath}
\usepackage{amssymb}
\usepackage{fixltx2e}
\usepackage{graphicx}
\usepackage{setspace}
\usepackage[authoryear]{natbib}
\usepackage{Tabbing}
\newcommand{\ra}[1]{\renewcommand{\arraystretch}{#1}}

\usepackage[unicode=true,pdfusetitle,
 bookmarks=true,bookmarksnumbered=false,bookmarksopen=false,
 breaklinks=true,pdfborder={0 0 0},backref=false,colorlinks=false]
 {hyperref}
\usepackage{breakurl}

\makeatletter

\providecommand{\tabularnewline}{\\}
\floatstyle{ruled}
\newfloat{algorithm}{tbp}{loa}
\providecommand{\algorithmname}{Algorithm}
\floatname{algorithm}{\protect\algorithmname}

  \theoremstyle{plain}
  \newtheorem{lem}{\protect\lemmaname}
  \theoremstyle{plain}
  
\theoremstyle{plain}
\newtheorem{thm}{\protect\theoremname}

\makeatother

  \providecommand{\lemmaname}{Lemma}
  \providecommand{\propositionname}{Proposition}
\providecommand{\theoremname}{Theorem}

\begin{document}
\global\long\def\balpha{\boldsymbol{\alpha}}
\global\long\def\bK{\mathbf{K}}
\global\long\def\bC{\mathbf{C}}
\global\long\def\bI{\mathbf{I}}

\newcommand{\di}{{\,\mathrm{d}}}
\global\long\def\bx{\mathbf{x}}
\global\long\def\bX{\mathbf{X}}
\global\long\def\argmin{\arg\!\min}
\global\long\def\argmax{\arg\!\max}
\global\long\def\bbeta{\boldsymbol{\beta}}
\global\long\def\by{\mathbf{x^{\prime}}}

\title{\LARGE \textbf{ Flexible Expectile Regression in Reproducing Kernel Hilbert Spaces}}
\author{Yi Yang \thanks{Department of Mathematics and Statistics, McGill University},
Teng Zhang \thanks{Princeton University. Yi and Teng are joint first authors.},
Hui Zou \thanks{Corresponding author, School of Statistics, University of Minnesota (zouxx019@umn.edu)}}

\date{\today}

\maketitle

\begin{abstract}
Expectile, first introduced by Newey and Powell (1987) in the econometrics literature, has recently become increasingly popular in risk management and capital allocation for financial institutions due to its desirable properties such as coherence and elicitability. The current standard tool for expectile regression analysis is the multiple linear expectile regression proposed by Newey and Powell in 1987. The growing applications of expectile regression motivate us to develop a much more flexible nonparametric multiple expectile regression in a reproducing kernel Hilbert space. The resulting estimator is called KERE which has multiple advantages over the classical multiple linear expectile regression by incorporating non-linearity, non-additivity and complex interactions in the final estimator. The kernel learning theory of KERE is established. We develop an efficient algorithm inspired by majorization-minimization principle for solving the entire solution path of KERE. It is shown that the algorithm converges at least at a linear rate. 
Extensive simulations are conducted to show the very competitive finite sample performance of KERE. We further demonstrate the application of KERE by using personal computer price data.
\end{abstract}

\section{Introduction}

The \emph{expectile} introduced by \citet{Asymmetric_Powell} is becoming
an increasingly popular tool in risk management and capital allocation
for financial institutions. Let $Y$ be a random variable, 
the $\omega$-expectile of $Y$, denoted as $f_{\omega}$, is defined by
\begin{equation}
\omega=\frac{E\{|Y-f_{\omega}|I_{Y\leq f_{\omega}}\}}{E\{|Y-f_{\omega}|\}},
\qquad\omega\in(0,1).\label{eq:def1}
\end{equation}
In financial applications,  the expectile has been widely used as a tool for efficient estimation of the expected shortfall (ES) through a one-one mapping between the two \citep{taylor2008estimating,hamidi2014dynamic,xie2014varying}. More recently, many researchers started to advocate the use of the expectile as a favorable
alternative
to other two commonly used risk measures -- Value at Risk (VaR) and
ES, due to its desirable properties such as \emph{coherence} and \emph{elicitability} \citep{kuan2009assessing, gneiting2011making,ziegel2014coherence}.  VaR has been criticized mainly for two drawbacks: First, it does not reflect the magnitude of the extreme losses for the underlying risk as it is only determined by the probability of such losses; Second, VaR is not a coherent risk measure due to the lack of the \emph{sub-additive} property \citep{emmer2013best,embrechts2014academic}. Hence the risk of merging portfolios together could get worse than
adding the risks separately, which contradicts the notion that risk
can be reduced by diversification \citep{artzner1999coherent}. Unlike VaR, ES is coherent and it considers the
magnitude of the losses when the VaR is exceeded. However, a major problem with ES is that it cannot be reliably backtested in the sense
that competing forecasts of ES cannot be properly evaluated through
comparison with realized observations. \citet{gneiting2011making}
attributed this weakness to the fact that ES does not have \emph{elicitability}. \citet{ziegel2014coherence}
further showed that the expectile are the only risk measure that is both coherent and elicitable.

In applications we often need to estimate the conditional expectile of the response variable given a set of covariates. This is called expectile regression. Statisticians and Econometricians pioneered the study of expectile regression. Theoretical properties of the multiple linear expectile were firstly studied in \citet{Asymmetric_Powell} and \citet{Asymmetric_Efron}.
\citet{Asymmetric_Tong} studied a non-parametric estimator of conditional expectiles based on local linear polynomials with a one-dimensional covariate, and established the asymptotic property of the estimator.  A semiparametric expectile regression model relying on penalized splines is proposed by \citet{sobotka2012geoadditive}.
\citet{yang2015nonparametric} adopted the gradient tree boosting algorithm for expectile regression.

In this paper, we propose a flexible nonparametric expectile regression estimator constructed in a reproducing kernel Hilbert space (RKHS) \citep{wahba}.
Our contributions in this article are twofold: First, we extend the parametric expectile model to a fully nonparametric multiple regression setting and develop the corresponding kernel learning theory. Second, we propose an efficient algorithm that adopts the Majorization-Minimization principle for computing the entire solution path of  the kernel expectile regression. We provide  numerical convergence analysis for the algorithm. Moreover, we provide an accompanying R package that allows other researchers and practitioners to use the kernel expectile regression.

The rest of the paper is organized as follows. In Section 2 we present the kernel expectile regression and develop an asymptotic learning theory. Section 3 derives the fast algorithm for solving the solution paths of the kernel expectile regression.  The numerical convergence of the algorithm is examined.  In Section 4 we use simulation models to show the high prediction accuracy of the kernel expectile regression. We analyze the personal computer price data in Section 5. The technical proofs are relegated to an appendix.

\section{Kernel Expectile Regression}
\subsection{Methodology}

\citet{Asymmetric_Powell}  showed that the $\omega$-expectile
$f_{\omega}$ of  $Y$ has an equivalent definition given by
\begin{equation}
f_{\omega}=\underset{f}{\arg\min}E\{\phi_{\omega}(Y-f) \},
\end{equation}
where
\begin{eqnarray}\label{eqdef1}
\phi_{\omega}(t)=\begin{cases}
(1-\omega)t^{2} & t\leq 0,\\
\omega t^{2} & t>0.
\end{cases}
\end{eqnarray}
Consequently, \citet{Asymmetric_Powell} showed that the $\omega$-expectile
$f_{\omega}$ of $Y$ given the set of covariates
$X=\mathbf{x}$, denoted by $f_{\omega}(\mathbf{x})$, can be defined as
\begin{equation}\label{eqdef2}
f_{\omega}(\mathbf{x})=\underset{f}{\arg\min}E\{\phi_{\omega}(Y-f)\mid X=\mathbf{x}\}.
\end{equation}
\citet{Asymmetric_Powell}  developed the multiple linear expectile regression based on \eqref{eqdef2}.
Given $n$ random observations
$(\bx_{1},y_{1}),\cdots,(\bx_{n},y_{n})$ with $\bx_{i}\in\mathbb{R}^{p}$
and $y_{i}\in\mathbb{R}$,
\citet{Asymmetric_Powell}  proposed the following formulation:
\begin{equation}
(\hat{\boldsymbol{\beta}},\hat{\beta}_0)=\underset{(\boldsymbol{\beta},\beta_0)}{\arg\min}\frac{1}{n}\sum_{i=1}^{n}\phi_{\omega}(y_{i}-\mathbf{x}_i^{\intercal}\boldsymbol{\beta}-\beta_0).\label{eq:optlinear}
\end{equation}
Then the estimated conditional  $\omega$-expectile is $\mathbf{x}_i^{\intercal}\hat{\boldsymbol{\beta}}+\hat{ \beta_0}.$
\citet{Asymmetric_Efron} proposed an efficient algorithm for computing \eqref{eq:optlinear}.

The linear expectile estimator can be too restrictive in many real applications. Researchers have also considered more flexible expectile regression estimators. For example, \citet{Asymmetric_Tong} studied a local linear-polynomial expectile estimator with a one-dimensional covariate.
However, the local fitting approach is not suitable  when the dimension of explanatory variables is more than five. This limitation of local smoothing motivated \citet{yang2015nonparametric} to develop a nonparametric expectile regression estimator based on the gradient tree boosting algorithm. The tree-boosted expectile regression tries to minimize the empirical expectile loss:
\begin{equation}
\underset{f\in{\cal F}}{\min}\frac{1}{n}\sum_{i=1}^{n}\phi_{\omega}(y_{i}-f(\mathbf{x}_{i})),\label{eq:opt}
\end{equation}
 where each candidate function $f\in \mathcal F$ is assumed to be an ensemble of regression trees.

In this article, we consider another nonparametric approach to the multiple expectile regression. To motivate our method, let us first look at the special expectile regression with $\omega=0.5$. It is easy to see from \eqref{eqdef1} and \eqref{eqdef2} that if $\omega=0.5$, expectile regression actually reduces to ordinary conditional mean regression. A host of flexible regression methods have been well-studied for the conditional mean regression, such as generalized additive model, regression trees, boosted regression trees, and function estimation in a reproducing kernel Hilbert space (RKHS). \citet{friedman2009elements} provided excellent introductions to all these methods. In particular, mean regression in a RKHS has a long history and a rich success record \citep{wahba}. So in the present work we propose the kernel expectile regression in a RKHS.

Denote by $\mathbb{H}_{K}$ the Hilbert space generated by a positive definite kernel $K$. By the Mercer's theorem, kernel $K$ has an eigen-expansion $K(\mathbf{x},\mathbf{x^{\prime}})=\sum_{i=1}^{\infty}\nu_{i}\varphi_{i}(\mathbf{x})\varphi_{i}(\mathbf{x^{\prime}})$ with $\nu_{i}\geq0$ and  $\sum_{i=1}^{\infty}\nu_{i}^{2}<\infty$. The
function $f$ in $\mathbb{H}_{K}$ can be expressed as an expansion of these eigen-functions $f(\mathbf{x})=\sum_{i=1}^{\infty}c_{i}\varphi_{i}(\mathbf{x})$ with the kernel induced squared norm $\|f\|_{\mathbb{H}_{K}}^{2}\equiv\sum_{i=1}^{\infty}c_{i}^{2}/\nu_{i}<\infty.$ 
Some most widely used kernel functions
are
\begin{itemize}
\item Gaussian RBF kernel $K(\mathbf{x},\mathbf{x^{\prime}})=\exp\left(\frac{-\|\mathbf{x}-\mathbf{x^{\prime}}\|^{2}}{\sigma^{2}}\right),$
\item Sigmoidal kernel $K(\mathbf{x},\mathbf{x^{\prime}})=\tanh(\kappa\left\langle \mathbf{x},\mathbf{x^{\prime}}\right\rangle +\theta),$
\item Polynomial kernel $K(\mathbf{x},\mathbf{x^{\prime}})=(\left\langle \mathbf{x},\mathbf{x^{\prime}}\right\rangle +\theta)^d.$
\end{itemize}
Other kernels can be found in \citet{smola1998connection} and \citet{friedman2009elements}.

Given $n$ observations $\{(\bx_{i},y_{i})\}_{i=1}^n$, the kernel expectile regression estimator (KERE) is defined as
\begin{equation}
(\hat{f}_{n}(\mathbf{x}),\hat{\alpha}_{0})=\arg\min_{f\in\mathbb{H}_{K},\alpha_{0}\in\mathbb{R}}\sum_{i=1}^{n}\phi_{\omega}(y_{i}-\alpha_{0}-f(\bx_{i}))+\lambda\|f\|_{\mathbb{H}_{K}}^{2},\label{eq:setup1}
\end{equation}
where $\mathbf{x}_{i}\in\mathbb{R}^{p}$, $\alpha_{0}\in\mathbb{R}$. The estimated conditional  $\omega$-expectile is $\hat{\alpha}_0+\hat{f}_{n}(\mathbf{x}).$ Sometimes, one can absorb the intercept term into the nonparametric function $f$. We keep the intercept term in order to make a direct comparison to the multiple linear expectile regression.

 
Although \eqref{eq:setup1} is often an optimization problem in an infinite-dimensional space, depending on the choice of the kernel, the representer theorem \citep{wahba} ensures that the solution to \eqref{eq:setup1} always lies in a finite-dimensional subspace spanned by kernel functions on observational data, i.e.,
\begin{equation}
f(\bx)=\sum_{i=1}^{n}\alpha_{i}K(\bx_{i},\bx),\label{eq:sol}
\end{equation}
for some $\{\alpha_{i}\}_{i=1}^{n}\subset\mathbb{R}$.

By \eqref{eq:sol} and the reproducing property of RKHS \citep{wahba}  we have
\begin{equation}
\|f\|_{\mathbb{H}_{K}}^{2}=\sum_{i=1}^{n}\sum_{j=1}^{n}\alpha_{i}\alpha_{j}K(\bx_{i},\bx_{j}).\label{eq:sol2}
\end{equation}

Based on \eqref{eq:sol} and \eqref{eq:sol2} we can rewrite the minimization
problem \eqref{eq:setup1} in a finite-dimensional space
\begin{equation}
\{\hat{\alpha}_{i}\}_{i=0}^{n}=\arg\min_{\{\alpha_{i}\}_{i=0}^{n}}\sum_{i=1}^{n}\phi_{\omega}\left(y_{i}-\alpha_{0}-\sum_{j=1}^{n}\alpha_{j}K(\bx_{i},\bx_{j})\right)+\lambda\sum_{i=1}^{n}\sum_{j=1}^{n}\alpha_{i}\alpha_{j}K(\bx_{i},\bx_{j}).\label{eq:setup2}
\end{equation}
The corresponding KERE estimator is 
$ \hat{\alpha}_0+\sum_{i=1}^{n} \hat \alpha_{i}K(\bx_{i},\bx)$.


The computation of KERE is based on \eqref{eq:setup2} and we use both \eqref{eq:setup1} and \eqref{eq:setup2} for the theoretical analysis of KERE.

\subsection{Kernel learning theory}

In this section we develop a kernel learning theory for KERE. We first discuss the criterion for evaluating an estimator in the context of expectile regression. Given the loss function $\phi_{\omega}$,  the risk is
$
{\cal R}(f,\alpha_{0})=E_{(\bx,y)}\phi_{\omega}(y-\alpha_{0}-f(\bx)).
$
It is argued that ${\cal R}(f,\alpha_{0})$ is a more appropriate evaluation measure in practice than the squared error risk defined as $E_{\bx}\|f(\bx)+\alpha_0-f^*_{\omega}(\bx) \|^2$, where $f^*_{\omega}$ is the true conditional expectile of $Y$ given $X=\bx$. The reason is simple:  Let $\hat f, \hat \alpha_0$ be any estimator based on the training data.  By law of large number we see that
\[
{\cal R}(\hat f,\hat \alpha_{0})=E_{\{y_j,\bx_j\}^m_{j=1}}\frac{1}{m}\sum^m_{j=1}\phi_{\omega}(y_j-\hat \alpha_{0}-\hat f(\bx_j)),
\]
and
\[
{\cal R}(\hat f,\hat \alpha_{0})=\lim_{m \rightarrow \infty}\frac{1}{m}\sum^m_{j=1}\phi_{\omega}(y_j-\hat \alpha_{0}-\hat f(\bx_j)),
\]
where $\{(\bx_j,y_j)\}^m_{j=1}$ is another independent test sample. Thus, one can use techniques such as cross-validation to estimate ${\cal R}(f,\alpha_{0})$. Additionally, the squared error risk depends on the function $f^*_{\omega}(\bx)$, which is usually unknown. Thus, we prefer to use ${\cal R}(\hat f,\hat \alpha_{0})$ over the squared error risk. Of course, if we assume 
a classical regression model (when $\omega=0.5$) such as $y=f(\bx)+\textrm{error}$, where the error is independent of $\bx$ with mean zero and constant variance, ${\cal R}(\hat f,\hat \alpha_{0})$ then just equals the squared error risk plus a constant. Unfortunately, such equivalence breaks down for other values of $\omega$ and more general models.

After choosing the risk function,  the goal is to minimize the risk. Since typically the estimation is done in a function space, the minimization is carried out in the chosen function space. In our case, the function space is RKHS generated by a kernel function $K$. Thus, the ideal risk is defined as
\[
{\cal R}^*_{f, \alpha_{0}}=\inf_{f\in\mathbb{H}_K,\alpha_0\in\mathbb{R}}{\cal R}(f, \alpha_{0}).
\]

Consider the kernel expectile regression estimator $(\hat{f},\hat{\alpha}_{0})$ as defined in \eqref{eq:setup1} based on a training sample $D_{n}$, where $D_{n}=\{(\bx_{i},y_{i})\}_{i=1}^n$ are i.i.d. drawn from an unknown distribution. The observed risk of KERE is
\[
{\cal R}(\hat{f},\hat\alpha_{0})=E_{(\bx,y)}\phi_{\omega}(y-\hat\alpha_{0}-\hat{f}(\bx)).
\]
It is desirable to show that ${\cal R}(\hat{f},\hat\alpha_{0})$ approaches the ideal risk ${\cal R}^*_{f, \alpha_{0}}$.

It is important to note that $ {\cal R}(\hat{f},\hat\alpha_{0})$ is a random quantity that depends on the training sample $D_{n}$. So it is not the usual risk function which is deterministic. However, we can consider the expectation of
$ {\cal R}(\hat{f},\hat\alpha_{0})$ and call it \emph{expected observed risk}. The formal definition is given as follows
\begin{equation}
\text{Expected observed risk:} \quad E_{D_n}{\cal R}(\hat{f},\hat{\alpha}_{0})=E_{D_{n}}\big\{E_{(\bx,y)}\phi_{\omega}(y-\hat{\alpha}_{0}-\hat{f}(\bx))\big\}.\label{eq:expobrisk}
\end{equation}

Our goal is to show that ${\cal R}(\hat{f},\hat\alpha_{0})$ converges to ${\cal R}^*_{f, \alpha_{0}}.$ We achieve this by showing that the expected observed risk converges to the ideal risk, i.e.,
$\lim_{n \rightarrow \infty } E_{D_n}{\cal R}(\hat{f},\hat{\alpha}_{0})={\cal R}^*_{f, \alpha_{0}}$.
By definition, we always have
${\cal R}(\hat{f},\hat\alpha_{0}) \ge {\cal R}^*_{f, \alpha_{0}}.
$ Then by Markov inequality, for any $\varepsilon>0$
$$
P\Big({\cal R}(\hat{f},\hat{\alpha}_{0}) - {\cal R}^*_{f, \alpha_{0}} >\varepsilon \Big) \le \frac{E_{D_n}{\cal R}(\hat{f},\hat{\alpha}_{0})- {\cal R}^*_{f, \alpha_{0}} }{\varepsilon} \rightarrow 0.
$$

The rigorous statement of our result is as follows:
\begin{thm}\label{thm:asymptotic}
Let $M=\sup_{\bx}K(\bx,\bx)^{1/2}$. 
Assume $M<\infty$ and $E y^2<D<\infty$ where $M$ and $D$ are two constants.
If $\lambda$ is chosen such that
as $n\rightarrow\infty$, $\lambda/n^{2/3}\rightarrow\infty$, $\lambda/n\rightarrow 0$, then we have
\[
E_{D_n}{\cal R}(\hat{f},\hat{\alpha}_{0})\rightarrow {\cal R}^*_{f, \alpha_{0}} \quad\text{as $n\rightarrow\infty$},
\]
and hence
\[
{\cal R}(\hat{f},\hat{\alpha}_{0}) - {\cal R}^*_{f, \alpha_{0}}  \rightarrow 0 \ \textrm{in probability}.
\]
\end{thm}

The Gaussian kernel is perhaps the most popular kernel for nonlinear learning. For the Gaussian kernel $K(\bx,\by)=\exp(-\|\bx-\by\|^2/c)$, we have $M=1$. For any radial kernel with the form
$K(\bx,\by)=h(\|\bx-\by\|)$ where $h$ is a smooth decreasing function, we see $M=h(0)^{\frac{1}{2}}$ which is finite as long as $h(0)<\infty$.

\section{Algorithm\label{sec:algorithm}}

\subsection{Derivation}
Majorization-minimization (MM) algorithm is a very successful technique for solving a wide range of statistical models \citep{lange2000optimization,hunter2004tutorial,MM08,zhou2010mm,lange2014mm}. In this section, we develop an algorithm inspired by  MM principle for solving the optimization problem \eqref{eq:setup2}.  Note that the loss function $\phi_{\omega}$ in \eqref{eq:setup2}
does not have the second derivative. We adopt the MM principle to find the minimizer by
iteratively minimizing a surrogate function that majorizes the objective
function  in \eqref{eq:setup2}.

To further simplify the notation we write $\balpha=(\alpha_{0},\alpha_{1},\alpha_{2},\cdots,\alpha_{n})^{\intercal}$,
and
\[
\bK_{i}=\left(1,K(\bx_{i},\bx_{1}),\ldots,K(\bx_{i},\bx_{n})\right),
\qquad\bK=\left(\begin{array}{ccc}
K(\bx_{1},\bx_{1}) & \cdots & K(\bx_{1},\bx_{n})\\
\vdots & \ddots & \vdots\\
K(\bx_{n},\bx_{1}) & \cdots & K(\bx_{n},\bx_{n})
\end{array}\right),
\]
\[
\bK_{0}=\left(\begin{array}{cc}
0 & \mathbf{0}_{n\times1}^{\intercal}\\
\mathbf{0}_{n\times1} & \mathbf{K}
\end{array}\right).
\]
Then \eqref{eq:setup2} is simplified to a minimization problem as
\begin{equation}
\widehat{\balpha}=\argmin_{\balpha}F_{\omega, \lambda}(\balpha),\label{eq:obj}
\end{equation}
\begin{equation}
F_{\omega, \lambda}(\balpha)=\sum_{i=1}^{n}\phi_{\omega}\left(y_{i}-\bK_{i}\balpha\right)+\lambda\balpha^{\intercal}\bK_{0}\balpha,\label{eq:floss}
\end{equation}
where $\omega$ is given for computing the corresponding level of the conditional expectile. We also assume that $\lambda$ is given for the time being. A smart algorithm for computing the solution for a sequence of $\lambda$ will be studied in Section \ref{sec:implementation}.

Our approach is to minimize \eqref{eq:obj} by
iteratively update $\balpha$ using the minimizer of a majorization function of $F_{\omega, \lambda}(\balpha)$. Specifically, at the $k$-th step of the algorithm,
where $k=0,1,2,\ldots$, assume that $\boldsymbol{\alpha}^{(k)}$
is the current value of $\balpha$ at iteration $k$, we find a majorization
function $Q(\balpha\mid\balpha^{(k)})$ for $F_{\omega, \lambda}(\balpha)$ at current  $\balpha^{(k)}$ that satisfies
\begin{alignat}{1}
Q(\balpha\mid\balpha^{(k)}) > & F_{\omega, \lambda}(\balpha)\quad\mathrm{when\ }\balpha \neq \balpha^{(k)},\label{eq:mathdef1}\\
Q(\balpha\mid\balpha^{(k)}) = & F_{\omega, \lambda}(\balpha) \quad\mathrm{when\ }\balpha = \balpha^{(k)}.\label{eq:mathdef2}
\end{alignat}
Then we update $\balpha$ by minimizing $Q(\balpha\mid\balpha^{(k)})$
rather than the actual objective function $F_{\omega, \lambda}(\balpha)$:
\begin{equation}
\balpha^{(k+1)}=\argmin_{\balpha}Q(\balpha\mid\balpha^{(k)}).\label{eq:majobj}
\end{equation}
To construct the majorization function $Q(\balpha\mid\balpha^{(k)})$
for $F_{\omega, \lambda}(\balpha)$ at the $k$-th iteration, we use the following lemma:
\begin{lem}
\label{lem:lips}The expectile loss $\phi_{\omega}$ has a Lipschitz continuous derivative $\phi^{\prime}_{\omega}$, i.e.
\begin{equation}
|\phi^{\prime}_{\omega}(a)-\phi^{\prime}_{\omega}(b)|\leq L|a-b|\qquad\forall a,b\in\mathbb{R},\label{eq:lips}
\end{equation}
where $L=2\max(1-\omega,\omega)$. This further implies that $\phi_{\omega}$ has a quadratic upper bound
\begin{equation}
\phi_{\omega}(a)\le\phi_{\omega}(b)+\phi^{\prime}_{\omega}(b)(a-b)+\frac{L}{2}|a-b|^{2}\qquad\forall a,b\in\mathbb{R}.\label{eq:upper_bound}
\end{equation}
Note that ``$=$'' is taken only  when  $a=b$.
\end{lem}

Assume the current ``residual" is $r_{i}^{(k)}=y_{i}-\bK_{i}\boldsymbol{\alpha}^{(k)}$,
then it is equivalent in \eqref{eq:obj} that  $y_{i}-\bK_{i}\balpha=r_{i}^{(k)}-\bK_{i}(\balpha-\balpha^{(k)})$.
By lemma \ref{lem:lips}, we obtain
\[
|\phi_{\omega}^{\prime}(r_{i}^{(k)}-\bK_{i}(\balpha-\balpha^{(k)}))-\phi_{\omega}^{\prime}(r_{i}^{(k)})|\leq 2\max(1-\omega,\omega)|\bK_{i}(\balpha-\balpha^{(k)})|,
\]
and the quadratic upper bound
\[
\phi_{\omega}(r_{i}^{(k)}-\bK_{i}(\balpha-\balpha^{(k)}))\leq q_{i}(\balpha\mid\balpha^{(k)}),
\]
where
\[
q_{i}(\balpha\mid\balpha^{(k)})=\phi_{\omega}(r_{i}^{(k)})-\phi_{\omega}^{\prime}(r_{i}^{(k)})\bK_{i}(\balpha-\balpha^{(k)})+\max(1-\omega,\omega)(\balpha-\balpha^{(k)})^{\intercal}\bK_{i}\bK_{i}^{\intercal}(\balpha-\balpha^{(k)}).
\]
Therefore the majorization function of $F_{\omega, \lambda}(\balpha)$ can be written
as
\begin{equation}
Q(\balpha\mid\balpha^{(k)})=\sum_{i=1}^{n}q_{i}(\balpha\mid\balpha^{(k)})+\lambda\balpha^{\intercal}\bK_{0}\balpha,\label{eq:qbound}
\end{equation}
which has an alternatively form that can be written as
\begin{equation}
Q(\balpha\mid\balpha^{(k)})=F_{\omega, \lambda}(\balpha^{(k)})+\nabla F_{\omega, \lambda}(\balpha^{(k)})(\balpha-\balpha^{(k)})+(\balpha-\balpha^{(k)})^{\intercal}\bK_{u}(\balpha-\balpha^{(k)}),\label{eq:alter_qbound}
\end{equation}
where
\begin{align}
\bK_{u} & =\lambda\bK_{0}+\max(1-\omega,\omega)\sum_{i=1}^{n}\bK_{i}\bK_{i}^{\intercal}\label{eq:ku}\\
 & =\max(1-\omega,\omega)\left(\begin{array}{cc}
n & \mathbf{1}^{\intercal}\mathbf{K}\\
\mathbf{K}\mathbf{1} & \mathbf{KK}+\frac{\lambda}{\max(1-\omega,\omega)}\mathbf{K}
\end{array}\right),
\end{align}
and $\mathbf{1}$ is an $n\times1$ vector of all ones.
Our algorithm updates $\balpha$ using the minimizer of the quadratic majorization
function \eqref{eq:alter_qbound}:
\begin{equation}
\balpha^{(k+1)}=\argmin_{\balpha}Q(\balpha\mid\balpha^{(k)})=\balpha^{(k)}+\bK_{u}^{-1}\left(-\lambda\bK_{0}\balpha^{(k)}+\frac{1}{2}\sum_{i=1}^{n}\phi_{\omega}^{\prime}(r_{i}^{(k)})\bK_{i}\right).\label{eq:update}
\end{equation}
The details of the whole procedures for solving \eqref{eq:obj} are described
in Algorithm \ref{alg:main}.
\begin{algorithm}
\caption{The algorithm for the minimization of \eqref{eq:obj}. \label{alg:main}}
\begin{itemize}
\item Let $\{y_{i}\}_{1}^{n}$ be observations of the response, $\{K(\bx_{i},\bx_{j})\}_{i,j=1}^{n}$
be the kernel of all observations, and $\balpha:=(\alpha_{0},\alpha_{1},\alpha_{2},\ldots,\alpha_{n})$.
\item Initialize $\balpha^{(0)}$ and $k=0$.
\item Iterate step 1--3 until convergence:\end{itemize}
\begin{enumerate}
\item Calculated the residue of the response by $r_{i}^{(k)}=y_{i}-\mathbf{K}_{i}\balpha^{(k)}$
for all $1\leq i\leq n$.
\item Obtain $\balpha^{(k+1)}$ by:
\[
\balpha^{(k+1)}=\balpha^{(k)}+\bK_{u}^{-1}\left(-\lambda\bK_{0}\balpha^{(k)}+\frac{1}{2}\sum_{i=1}^{n}\phi_{\omega}^{\prime}(r_{i}^{(k)})\bK_{i}\right),
\]
where
\[
\bK_{u}=\max(1-\omega,\omega)\left(\begin{array}{cc}
n & \mathbf{1}^{\intercal}\mathbf{K}\\
\mathbf{K}\mathbf{1} & \mathbf{KK}+\frac{\lambda}{\max(1-\omega,\omega)}\mathbf{K}
\end{array}\right).
\]
\item $k:=k+1$.
\end{enumerate}
\end{algorithm}

\subsection{Convergence analysis}\label{sec:convergence_analysis}

Now we provide the convergence analysis of Algorithm \ref{alg:main}. Lemma \ref{lem:convergence}  below shows that the sequence $(\balpha^{(k)})$ in the algorithm converges to the unique global minimum $\widehat{\balpha}$ of the optimization problem.
\begin{lem}\label{lem:convergence}
If we update $\balpha^{(k+1)}$ by using \textup{\eqref{eq:update}},
then the following results hold:
\begin{enumerate}
\item The descent property of the objective function. $F_{\omega, \lambda}(\balpha^{(k+1)})\leq F_{\omega, \lambda}(\balpha^{(k)})$,
$\forall k$.
\item The convergence of $\balpha$. Assume that $\sum_{i=1}^{n}\bK_{i}\bK_{i}^{\intercal}$
is a positive definite matrix, then $\lim_{k\rightarrow\infty}\|\balpha^{(k+1)}-\balpha^{(k)}\|=0$.
\item The sequence $(\balpha^{(k)})$ converges to $\widehat{\balpha}$,
which is the unique global minimum of \eqref{eq:obj}.
\end{enumerate}
\end{lem}

\begin{thm}
\label{thm:iteration} Denote by $\widehat{\balpha}$ the unique minimizer
of \eqref{eq:obj} and
\begin{equation}
\Lambda_{k}=\frac{Q(\widehat{\balpha}\mid\balpha^{(k)})-F_{\omega, \lambda}(\widehat{\balpha})}{(\widehat{\balpha}-\balpha^{(k)})^{\intercal}\bK_{u}(\widehat{\balpha}-\balpha^{(k)})}.\label{eq:LambdaK}
\end{equation}
Note that when $\Lambda_{k}=0$, it is just a trivial case $\balpha^{(j)}=\widehat{\balpha}$
for $j>k$. We define
\[
\Gamma=1-\gamma_{\min}(\bK_{u}^{-1}\bK_{l}),
\]
where
\[
\bK_{l}=\lambda\bK_{0}+\min(1-\omega,\omega)\sum_{i=1}^{n}\bK_{i}\bK_{i}^{\intercal}.
\]
Assume that $\sum_{i=1}^{n}\bK_{i}\bK_{i}^{\intercal}$ is a positive
definite matrix. Then we have the following results:

1. \textup{$F_{\omega, \lambda}(\balpha^{(k+1)})-F_{\omega, \lambda}(\widehat{\balpha})\leq\Lambda_{k}\left(F_{\omega, \lambda}(\balpha^{(k)})-F_{\omega, \lambda}(\widehat{\balpha})\right).$}

2. The sequence $(F_{\omega, \lambda}(\balpha^{(k)}))$ has a linear convergence rate
no greater than $\Gamma$, and $0\leq\Lambda_{k}\leq\Gamma<1$.

3. The sequence $(\balpha^{(k)})$ has a linear convergence rate no
greater than $\sqrt{\Gamma\gamma_{\max}(\mathbf{K}_{u})/\gamma_{\min}(\mathbf{K}_{l})}$, i.e.
\[
\|\balpha^{(k+1)}-\widehat{\balpha}\|\leq \sqrt{\Gamma\frac{\gamma_{\max}(\mathbf{K}_{u})}{\gamma_{\min}(\mathbf{K}_{l})}}\|\balpha^{(k)}-\widehat{\balpha}\|.
\]
\end{thm}

Theorem~\ref{thm:iteration} says that the convergence rate of Algorithm \ref{alg:main} is at least linear. In our numeric experiments, we have found that Algorithm \ref{alg:main}
converges very fast: the convergence criterion is usually met after 15 iterations.

\subsection{Implementation}\label{sec:implementation}
We discuss some techniques used in our implementation to  further improve the computational speed of the algorithm.

Usually expectile models are computed by applying Algorithm \ref{alg:main}
on a descending sequence of $\lambda$ values. To create a
sequence $\{\lambda_{m}\}_{m=1}^{M}$, we place $M-2$ points uniformly (in the log-scale) between
the starting and ending point $\lambda_{\max}$ and $\lambda_{\min}$ such
that the $\lambda$ sequence length is $M$. The default number
for $M$ is 100, hence $\lambda_{1}=\lambda_{\max}$, and $\lambda_{100}=\lambda_{\min}$. We adopt the warm-start trick to implement the solution
paths along $\lambda$ values: suppose that we have already obtained
the solution $\widehat{\boldsymbol{\alpha}}_{\lambda_{m}}$ at $\lambda_{m}$,
then $\widehat{\boldsymbol{\alpha}}_{\lambda_{m}}$ will be used as
the initial value for computing the solution at $\lambda_{m+1}$ in
Algorithm \ref{alg:main}.

Another computational trick adopted is based on the fact that in Algorithm \ref{alg:main},  the inverse of $\mathbf{K}_{u}$ does not have to be re-calculated for each $\lambda$. There is an easy way to update $\mathbf{K}^{-1}_{u}$ for $\lambda_{1},\lambda_{2},\ldots$. Because $\mathbf{K}_{u}$ can be partitioned into two rows and two
columns of submatrices, by Theorem 8.5.11 of \citet{harville2008matrix},
$\mathbf{K}^{-1}_{u}$ can be expressed as

\begin{align}
\bK_{u}^{-1}(\lambda) & =\frac{1}{\max(1-\omega,\omega)}\left(\begin{array}{cc}
n & \mathbf{1}^{\intercal}\mathbf{K}\\
\mathbf{K}\mathbf{1} & \mathbf{KK}+\frac{\lambda}{\max(1-\omega,\omega)}\mathbf{K}
\end{array}\right)^{-1}\nonumber \\
 & =\frac{1}{\max(1-\omega,\omega)}\left[\left(\begin{array}{cc}
\frac{1}{n} & \mathbf{0}_{1\times n}\\
\mathbf{0}_{n\times1} & \mathbf{0}_{n\times n}
\end{array}\right)+\left(\begin{array}{c}
-\frac{1}{n}\mathbf{1}^{\intercal}\mathbf{K}\\
\mathbf{I}_{n}
\end{array}\right)\mathbf{Q}_{\lambda}^{-1}(-\frac{1}{n}\mathbf{K}\mathbf{1},\mathbf{I}_{n})\right],\label{eq:Ku_partition}
\end{align}
where
\[
\mathbf{Q}_{\lambda}^{-1}=\left[\left(\mathbf{KK}+\frac{\lambda}{\max(1-\omega,\omega)}\mathbf{K}\right)-\frac{1}{n}\mathbf{K}\mathbf{1}\mathbf{1}^{\intercal}\mathbf{K}\right]^{-1}.
\]
In \eqref{eq:Ku_partition} only $\mathbf{Q}_{\lambda}^{-1}$ changes
with $\lambda$, therefore the computation of $\bK_{u}^{-1}$ for a different
$\lambda$ only requires the updating of $\mathbf{Q}_{\lambda}^{-1}$.
Observing that $\mathbf{Q}_{\lambda}^{-1}$ is the inverse of the
sum of two submatrices $\mathbf{A}$ and $\mathbf{B}$:
\[
\mathbf{A}_{\lambda}=\mathbf{KK}+\frac{\lambda}{\max(1-\omega,\omega)}\mathbf{K},\qquad\mathbf{B}=-\frac{1}{n}\mathbf{K}\mathbf{1}\mathbf{1}^{\intercal}\mathbf{K}.
\]
By Sherman\textendash{}Morrison formula \citep{1950},
\begin{equation}
\mathbf{Q}_{\lambda}^{-1}  =\left[\mathbf{A}_{\lambda}+\mathbf{B}\right]^{-1}=\mathbf{A}_{\lambda}^{-1}-\frac{1}{1+g}\mathbf{A}_{\lambda}^{-1}\mathbf{B}\mathbf{A_{\lambda}}^{-1},
\label{eq:qa_relation}
\end{equation}
where $g=\mathrm{trace}(\mathbf{B}\mathbf{A}_{\lambda}^{-1})$, we
find that to get $\mathbf{Q}_{\lambda}^{-1}$ for a different $\lambda$
one just needs to get $\mathbf{A}_{\lambda}^{-1}$, which can be efficiently
computed by using eigen-decomposition $\mathbf{\mathbf{K}=\mathbf{U}\mathbf{D}\mathbf{U}^{\intercal}}$:
\begin{equation}
\mathbf{A}_{\lambda}^{-1}=\left(\mathbf{KK}+\frac{\lambda}{\max(1-\omega,\omega)}\mathbf{K}\right)^{-1}=\mathbf{U}\left(\mathbf{D}^{2}+\frac{\lambda}{\max(1-\omega,\omega)}\mathbf{I}_{n}\right)^{-1}\mathbf{U}^{\intercal}.\label{eq:a_inv_update}
\end{equation}
\eqref{eq:a_inv_update} implies that the computation of $\bK_{u}^{-1}(\lambda)$ depends only on $\lambda$, $\mathbf{D}$, $\mathbf{U}$ and $\omega$. Since $\mathbf{D}$, $\mathbf{U}$ and $\omega$ stay unchanged, we only need to calculate them once. To get $\bK_{u}^{-1}(\lambda)$ for a different $\lambda$ in the sequence, we just need to plug in a new $\lambda$ in \eqref{eq:a_inv_update}.
 
The following is the implementation for computing KERE for a sequence of $\lambda$ values using Algorithm 1:
\begin{itemize}
\item Calculate $\mathbf{U}$ and $\mathbf{D}$ according to $\mathbf{\mathbf{K}=\mathbf{U}\mathbf{D}\mathbf{U}^{\intercal}}$.
\item Initialize $\ensuremath{\widehat{\balpha}_{\lambda_{0}}=[0,0,\ldots,0]}$.
\item \textbf{for} $m=1,2,\ldots,M$, repeat step 1-3:
\begin{enumerate}
\item Initialize $\boldsymbol{\alpha}_{\lambda_{m}}^{(0)}=\widehat{\boldsymbol{\alpha}}_{\lambda_{m-1}}$.
\item Compute $\bK_{u}^{-1}(\lambda_{m})$ using \eqref{eq:Ku_partition}, \eqref{eq:qa_relation} and \eqref{eq:a_inv_update}.
\item Call Algorithm \ref{alg:main} to compute $\widehat{\boldsymbol{\alpha}}_{\lambda_{m}}$.
\end{enumerate}
\end{itemize}

Our algorithm has been implemented in an official R package \texttt{KERE}, which is publicly available from the Comprehensive R Archive Network at \url{http://cran.r-project.org/web/packages/KERE/index.html}.

\section{Simulation}

In this section, we conduct extensive simulations to show the excellent finite performance of KERE.
We investigate how the performance of KERE is affected by various model and error distribution settings,
training sample sizes and other characteristics. Although many kernels are available, throughout
this section we use the commonly recommended \citep{friedman2009elements}
Gaussian radial basis function (RBF) kernel $K(\bx_{i},\mathbf{x}_{j})=e^{\frac{-\|\bx_{i}-\mathbf{x}_{j}\|^{2}}{\sigma^{2}}}$.
We select the best pair of kernel bandwidth $\sigma^{2}$ and regularization
parameter $\lambda$ by two-dimensional five-fold cross-validation. All computations
were done on an Intel Core i7-3770 processor at 3.40GHz.

\subsection*{Simulation I: single covariate case}

The model used for this simulation is defined as
\begin{equation}
y_{i}=\sin(0.7x_{i})+\frac{x_{i}^{2}}{20}+\frac{|x_{i}|+1}{5}\epsilon_{i},\label{eq:sim1}
\end{equation}
which is heteroscedastic as  error  depends on a single covariate
$x\sim U[-8,8]$. We used a single covariate such that the estimator can be visualized nicely.

We used two different error distributions: Laplace
distribution and a mixed normal distribution,
\[
\epsilon_{i}\sim0.5N(0,\frac{1}{4})+0.5N(1,\frac{1}{16}).
\]
We generated $n=400$ training observations from \eqref{eq:sim1},
on which five expectile models with levels $\omega=\{0.05,0.2,0.5,0.8,0.95\}$
were fitted. We selected the best $(\sigma^{2},\lambda)$ pair by
using two-dimensional, five-fold cross-validation. We generated an
additional $n^{\prime}=2000$ test observations for evaluating the
mean absolute deviation (MAD)  of the final estimate. Assume that the
true expectile function is $f_{\omega}$ and the predicted expectile
is $\hat{f}_{\omega}$, then the mean absolute deviation are defined
a
\[
\mathrm{MAD}(\omega)=\frac{1}{n^{\prime}}\sum_{i=1}^{n^{\prime}}|f_{\omega}(\mathbf{x}_{i})-\hat{f}_{\omega}(\mathbf{x}_{i})|.
\]
The true expectile $f_{\omega}$ is equal to $\sin(0.7x)+\frac{x^{2}}{20}+\frac{|x|+1}{5}b_{\omega}(\epsilon)$,
where $b_{\omega}(\epsilon)$ is the $\omega$-expectile of $\epsilon$, which is the theoretical minimizer of $E\phi_{\omega}(\epsilon-b)$.

The simulations were repeated for 100 times under the above settings.
We recorded MADs for different expectile levels in Table \ref{tab:simulation_MAD}.
We find that the accuracy of the expectile prediction with mixed normal
errors is generally better than that with Laplace errors. For the symmetric
Laplace case, the prediction MADs are also symmetric around $\omega=0.5$,
while for the skewed mixed-normal case the MADs are skewed. In order to
show that KERE works as expected, in Figure~\ref{fig:simulation1}
we also compared the theoretical and predicted expectile curves based on KERE with
$\omega=\{0.05,0.2,0.5,0.8,0.95\}$ in Figure \ref{fig:simulation1}.
We can see that the corresponding theoretical and predicted curves
are very close. Theoretically the two should become the same curves as
$n\rightarrow\infty$.

\begin{figure}
\centering{}\includegraphics[scale=0.55]{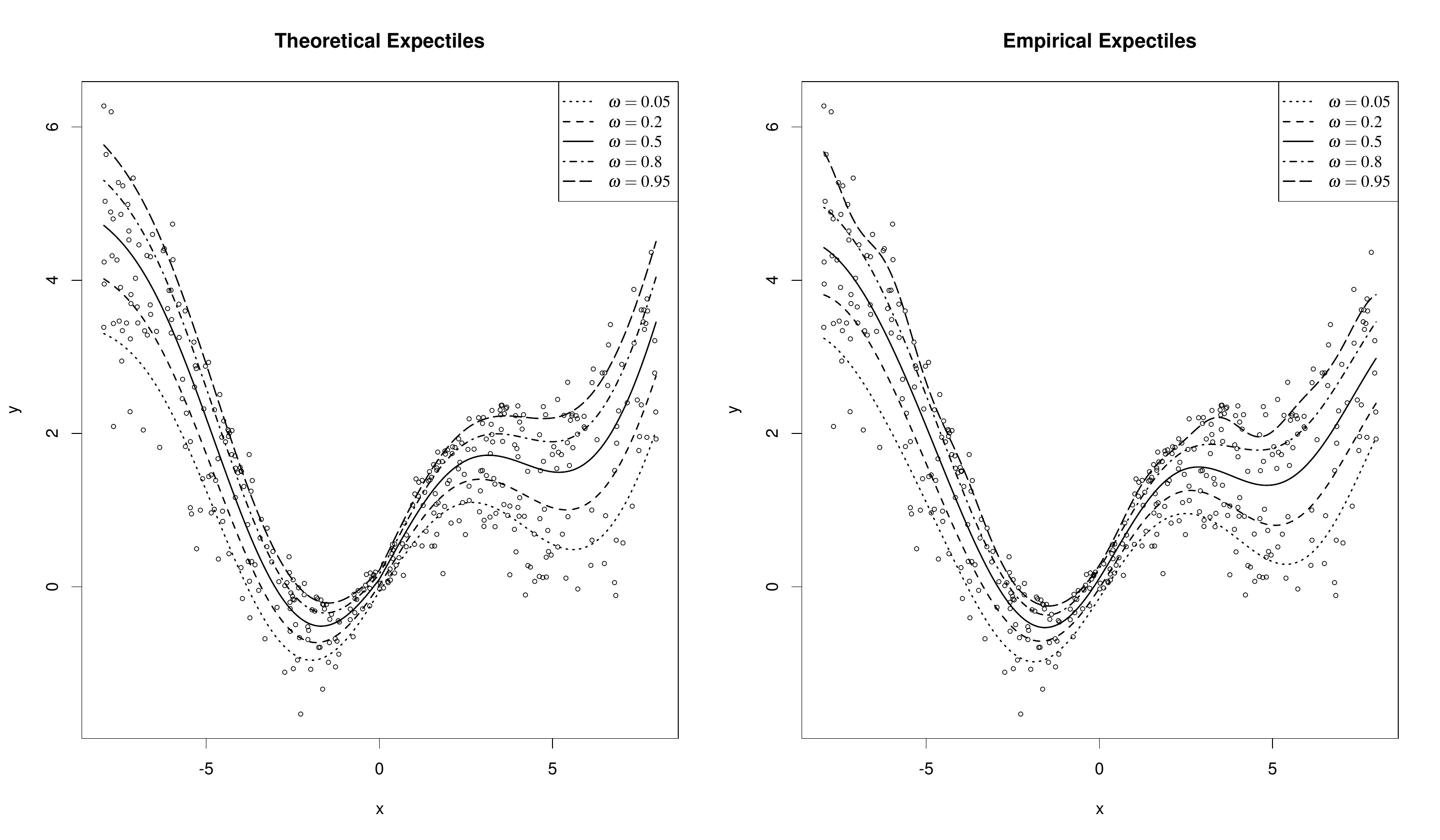} \caption{Theoretical expectiles and empirical expectiles for a covariate heteroscedastic
model with mixed normal error. The model is fitted on five expectile
levels $\omega=\{0.05,0.2,0.5,0.8,0.95\}$.\label{fig:simulation1}}
\end{figure}

\begin{table}
\ra{1.1}

\begin{centering}
\begin{tabular}{cccccc}
\toprule
$\omega$ & 0.05 & 0.2 & 0.5 & 0.8 & 0.95\tabularnewline
\midrule
Mixture & 0.236 (0.003) & 0.138 (0.003) & 0.376 (0.002) & 0.610 (0.002) & 0.788 (0.002)\tabularnewline
Laplace & 2.346 (0.013) & 1.037 (0.007) & 0.179 (0.005) & 1.033 (0.006) & 2.333 (0.027)\tabularnewline
\bottomrule
\end{tabular}
\par\end{centering}

\caption{The averaged MADs and the corresponding standard errors of expectile
regression predictions for single covariate heteroscedastic models
with mixed normal and Laplace error. The models are fitted on five
expectile levels $\omega=\{0.05,0.2,0.5,0.8,0.95\}$. The results
are based on 300 independent runs. \label{tab:simulation_MAD}}
\end{table}

\subsection*{Simulation II: multiple covariate case}

In this part we illustrate that KERE can work very well for target
functions that are non-additive and/or with complex interactions. We
generated data $\{\mathbf{x}_{i},y_{i}\}_{i=1}^{n}$  according to
\[
y_{i}=f_{1}(\mathbf{x}_{i})+|f_{2}(\mathbf{x}_{i})|\epsilon_{i},
\]
where predictors $\mathbf{x}_{i}$ was generated from a joint normal
distribution $N(0,\mathbf{I}_{p})$ with $p=10$. For the error term
$\epsilon_{i}$ we consider three types of distributions:
\begin{enumerate}
\item Normal distribution $\epsilon_{i}\sim N(0,1)$.
\item Student's $t$-distribution with four degrees of freedom $\epsilon_{i}\sim t_{4}$.
\item Mixed normal distribution $\epsilon_{i}\sim0.9N(0,1)+0.1N(1,4)$.
\end{enumerate}

We now describe the construction of $f_1$ and $f_2$.
In the homoscedastic model, we let $f_{2}(\mathbf{x}_{i})=1$ and $f_1$ is generated by the ``random function
generator'' model \citep{Friedman00greedyfunction}, according to
\[
f(\mathbf{x})=\sum_{l=1}^{20}a_{l}g_{l}(\mathbf{x}_{l}),
\]
where $\{a_{l}\}_{l=1}^{20}$ are sampled from uniform distribution
$a_{l}\sim U[-1,1]$, and $\mathbf{x}_{l}$ is a random subset of
$p$-dimensional predictor $\mathbf{x}$, with size $p_{l}=\min(\lfloor1.5+r,p\rfloor)$,
where $r$ was sampled from exponential distribution $r\sim Exp(0.5)$.
The function $g_{l}(\mathbf{x}_{l})$ is an $p_{l}$-dimensional Gaussian
function:
\[
g_{l}(x_{l})=\exp\Big[-\frac{1}{2}(\mathbf{x}_{l}-\boldsymbol{\mu}_{l})^{\intercal}\mathbf{V}_{l}(\mathbf{x}_{l}-\boldsymbol{\mu}_{l})\Big],
\]
where $\boldsymbol{\mu}_{l}$ follows the distribution $N(0,\mathbf{I}_{p_{l}})$.
The $p_{l}\times p_{l}$ covariance matrix $\mathbf{V}_{l}$ is defined
by $\mathbf{V}_{l}=\mathbf{U}_{l}\mathbf{D}_{l}\mathbf{U}_{l}^{\intercal}$,
where $\mathbf{U}_{l}$ is a random orthogonal matrix, and $\mathbf{D}_{l}=\mathrm{diag}(d_{1l},d_{2l},\cdots,d_{p_{l}l})$
with $\sqrt{d_{jl}}\sim U[0.1,2]$.

In the heteroscedastic model,  $f_{1}$ is the same as in the homoscedastic model and $f_{2}$ is independently generated by the ``random function
generator'' model.

We generated $n=300$ observations as the training set, on which the
estimated expectile functions $\hat{f}_{\omega}$ were computed at
seven levels:
\[
\omega\in\{0.05,0.1,0.25,0.5,0.75,0.9,0.95\}.
\]
An additional test set with $n^{\prime}=1200$ observations was generated
for evaluating MADs between the fitted expectile $\hat{f}_{\omega}$ and
the true expectile $f_{\omega}$. Note that the expectile
function $f_{\omega}(\mathbf{x})$ is equal to $f_{1}(\mathbf{x})+b_{\omega}(\epsilon)$
in the homoscedastic model and $f_{1}(\mathbf{x})+|f_{2}(\mathbf{x})|b_{\omega}(\epsilon)$
in the heteroscedastic model, where $b_{\omega}(\epsilon)$ is the
$\omega$-expectile of the error distribution. Under the above settings,
we repeated the simulations for 300 times and record the MAD and timing
each time.

In Figure \ref{fig:sim1} and \ref{fig:sim2} we show the box-plots of
empirical distributions of MADs, and in Table \ref{tab:homo_heter_mad}
we report the average values of $\mathrm{MADs}$ and corresponding
standard errors. We see that KERE can deliver accurate expectile prediction
results in all cases, although relatively the prediction error is
more volatile in the heteroscedastic case as expected: in the mean
regression case ($\omega=0.5$), the averaged MADs in homoscedastic
and heteroscedastic models are very close. But this difference  grows
larger as $\omega$ moves away from $0.5$. We also observe that the
prediction MADs for symmetric distributions, normal and $t_{4}$,
also appear to be symmetric around the conditional mean $\omega=0.5$,
and that the prediction MADs in the skewed mixed-normal distribution
cases are asymmetric. The total computation times for conducting two-dimensional,
five-fold cross-validation and fitting the final model with the chosen
parameters $(\sigma^{2},\lambda)$ for conditional expectiles are
also reported in Table \ref{tab:homo_heter_time}. We find that the
algorithm can efficiently solve all models under 20 seconds, regardless
of choices of error distributions.

\begin{figure}
\begin{centering}
\includegraphics[scale=0.65]{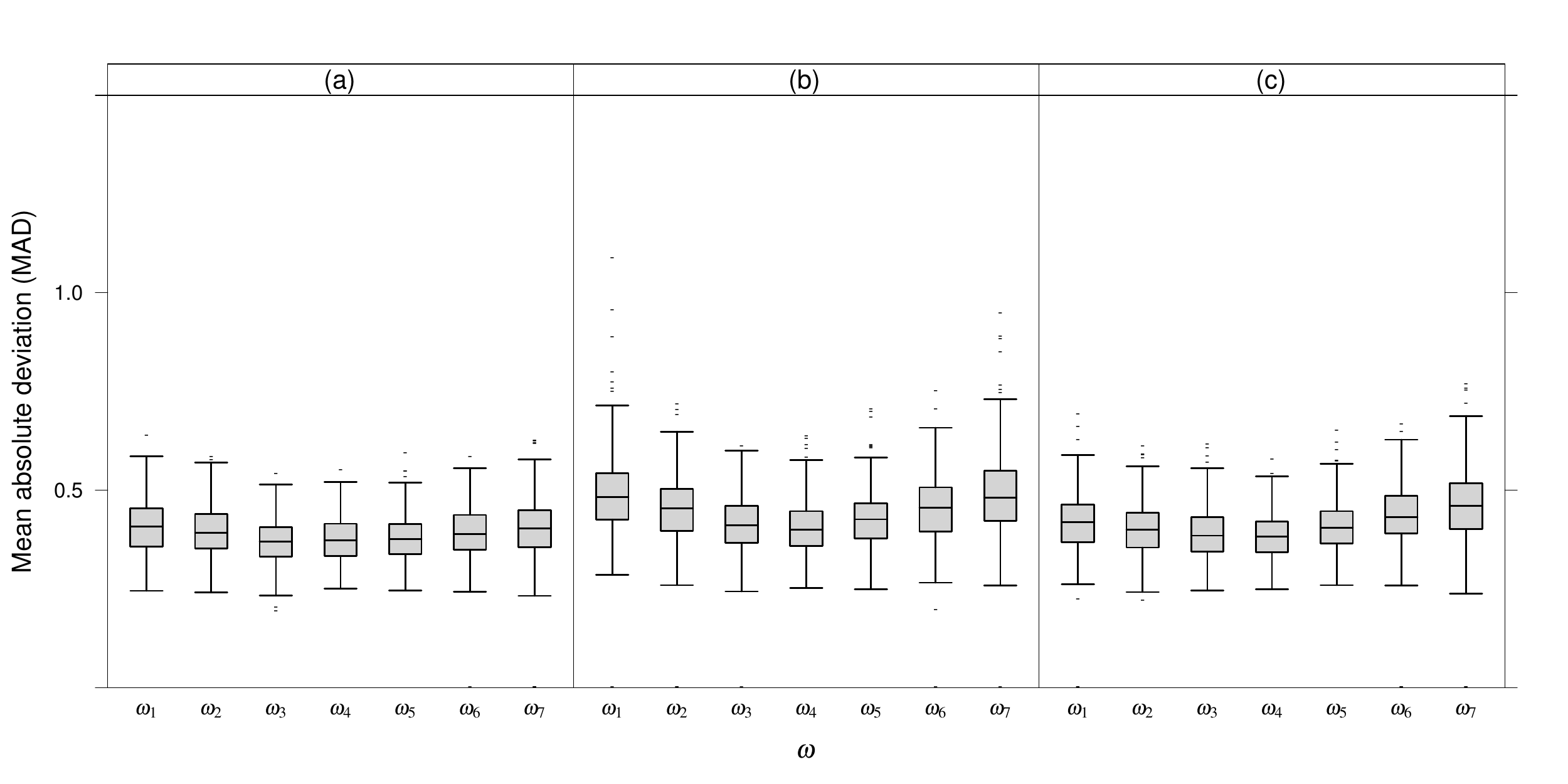}
\par\end{centering}

\caption{Homoscedastic models with error distribution (a) normal, (b) $t_{4}$
distribution, (c) mixed normal. Box-plots show MADs based on 300 independent
runs for expectiles $\omega\in\{0.05,0.1,0.25,0.5,0.75,0.9,0.95\}$.
\label{fig:sim1}}
\end{figure}

\begin{figure}
\begin{centering}
\includegraphics[scale=0.65]{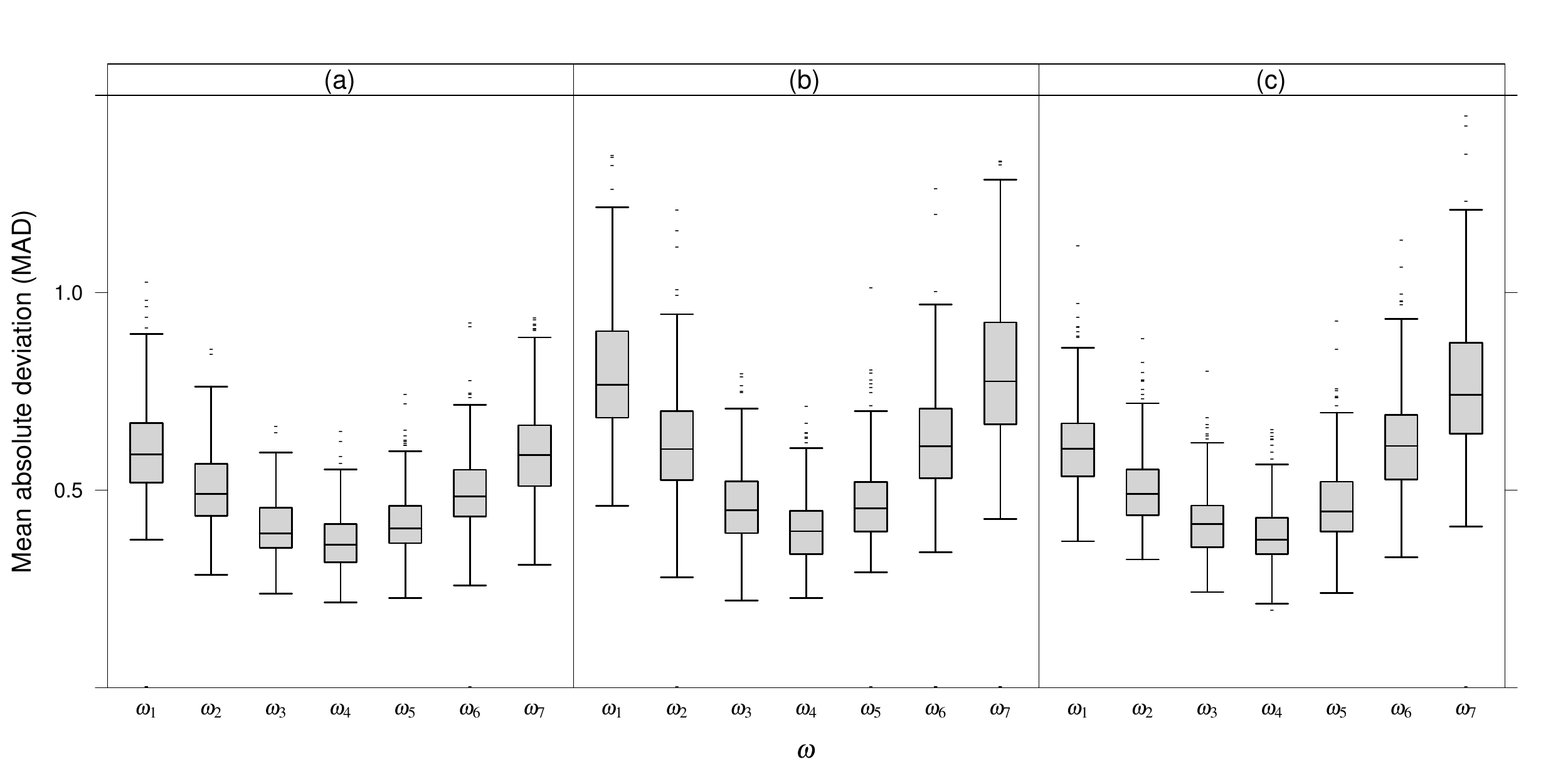}
\par\end{centering}

\caption{Heteroscedastic models with error distribution (a) normal, (b) $t_{4}$
distribution, (c) mixed normal. Box-plots show MADs based on 300 independent
runs for expectiles $\omega\in\{0.05,0.1,0.25,0.5,0.75,0.9,0.95\}$.
\label{fig:sim2}}
\end{figure}

\begin{table}
\ra{1.1}

\begin{centering}
\begin{tabular}{lrrrrrrr}
\toprule
 & \multicolumn{3}{c}{Homoscedastic model} & \phantom{}  & \multicolumn{3}{c}{Heteroscedastic model}\tabularnewline
\cmidrule{2-4} \cmidrule{6-8}
$\omega$ & Normal & $t_{4}$ & Mixture  &  & Normal & $t_{4}$ & Mixture\tabularnewline
\midrule
0.05 & 0.4068 & 0.4916 & 0.4183 &  & 0.6009 & 0.8035 & 0.6142\tabularnewline
 & (0.0039) & (0.0061) & (0.0046) &  & (0.0079) & (0.0103) & (0.0066)\tabularnewline
0.1 & 0.3975 & 0.4529 & 0.4019 &  & 0.5067 & 0.6315 & 0.5052\tabularnewline
 & (0.0037) & (0.0051) & (0.0037) &  & (0.0056) & (0.0094) & (0.0054)\tabularnewline
0.25 & 0.3717 & 0.4145 & 0.3886 &  & 0.4065 & 0.4648 & 0.4173\tabularnewline
 & (0.0031) & (0.0042) & (0.0038) &  & (0.0042) & (0.0061) & (0.0047)\tabularnewline
0.5 & 0.3750 & 0.4069 & 0.3851 &  & 0.3712 & 0.4038 & 0.3886\tabularnewline
 & (0.0032) & (0.0038) & (0.0032) &  & (0.0042) & (0.0049) & (0.0045)\tabularnewline
0.75 & 0.3782 & 0.4261 & 0.4102 &  & 0.4185 & 0.4702 & 0.4635\tabularnewline
 & (0.0033) & (0.0042) & (0.0036) &  & (0.0046) & (0.0064) & (0.0057)\tabularnewline
0.9 & 0.3932 & 0.4553 & 0.4356 &  & 0.4968 & 0.6226 & 0.6203\tabularnewline
 & (0.0038) & (0.0050) & (0.0045) &  & (0.0058) & (0.0081) & (0.0076)\tabularnewline
0.95 & 0.4040 & 0.4925 & 0.4628 &  & 0.5938 & 0.8078 & 0.7631\tabularnewline
 & (0.0046) & (0.0062) & (0.0054) &  & (0.0066) & (0.0128) & (0.0102)\tabularnewline
\bottomrule
\end{tabular}
\par\end{centering}

\caption{The averaged MADs and the corresponding standard errors for fitting
homoscedastic and heteroscedastic models based on 300 independent
runs. The expectile levels are $\omega\in\{0.05,0.1,0.25,0.5,0.75,0.9,0.95\}$.\label{tab:homo_heter_mad}}
\end{table}

\begin{table}
\ra{1.1}

\begin{centering}
\begin{tabular}{lrrrrrrr}
\toprule
 & \multicolumn{3}{c}{Homoscedastic model} & \phantom{}  & \multicolumn{3}{c}{Heteroscedastic model}\tabularnewline
\cmidrule{2-4} \cmidrule{6-8}
$\omega$ & Normal & $t_{4}$ & Mixture  &  & Normal & $t_{4}$ & Mixture\tabularnewline
\midrule
0.05 & 19.04 & 21.47 & 17.10 &  & 16.90 & 17.60 & 17.95\tabularnewline
0.1 & 14.25 & 16.89 & 13.91 &  & 14.38 & 14.60 & 15.21\tabularnewline
0.25 & 11.67 & 15.25 & 13.59 &  & 12.30 & 12.49 & 12.36\tabularnewline
0.5 & 10.54 & 14.09 & 12.18 &  & 10.92 & 11.13 & 11.01\tabularnewline
0.75 & 8.24 & 15.33 & 10.47 &  & 12.48 & 12.48 & 12.38\tabularnewline
0.9 & 10.08 & 14.39 & 12.46 &  & 14.67 & 15.25 & 14.52\tabularnewline
0.95 & 12.16 & 19.90 & 15.17 &  & 17.34 & 17.75 & 16.61\tabularnewline
\bottomrule
\end{tabular}
\par\end{centering}

\caption{The averaged computation times (in seconds) for fitting homoscedastic
and heteroscedastic models based on 300 independent runs. The expectile
levels are  $\omega\in\{0.05,0.1,0.25,0.5,0.75,0.9,0.95\}$.\label{tab:homo_heter_time}}
\end{table}

We next study how sample size affects predictive performance and computational
time. We fit expectile models with $\omega\in\{$0.1, 0.5, $0.9\}$
using various sizes of training sets $n\in\{$250, 500, 750, $1000\}$
and evaluate the prediction accuracy of the estimate using an independent
test set of size $n^{\prime}=2000$. We then report the averaged MADs
and the corresponding averaged timings in Table \ref{tab:sample_size_table}.
Since the results are very close for different model settings, only
the result from the heteroscedastic model with mixed-normal error
is presented. We find that the sample size strongly affects predictive
performance and timings: large samples give models with higher predictive
accuracy at the expense of computational cost -- the timings as least
quadruple as one doubles sample size.

\begin{table}
\ra{1.1}

\begin{centering}
\begin{tabular}{ccccccccccc}
\toprule
 &  & \multicolumn{4}{c}{Error} &  & \multicolumn{4}{c}{Timing}\tabularnewline
\cmidrule{3-6} \cmidrule{8-11}
$n$ &  & 250 & 500 & 750 & 1000 &  & 250 & 500 & 750 & 1000\tabularnewline
\midrule
$\omega=0.1$ &  & 0.4824 & 0.4084 & 0.4047 & 0.3887 &  & 8.739 & 56.188 & 168.636 & 382.897\tabularnewline
$\omega=0.5$ &  & 0.3329 & 0.2977 & 0.2732 & 0.2544 &  & 6.028 & 43.802 & 159.398 & 329.646\tabularnewline
$\omega=0.9$ &  & 0.6341 & 0.5861 & 0.5563 & 0.5059 &  & 9.167 & 56.533 & 173.359 & 386.345\tabularnewline
\bottomrule
\end{tabular}
\par\end{centering}

\caption{The averaged MADs and the corresponding averaged computation times
(in seconds) are reported. The size of the training set varies from
250 to 1000. The size of the test data set is 2000. All models are
fitted on three expectile levels: (a) $\omega=0.1$, (b) $\omega=0.5$
and (c) $\omega=0.9$. \label{tab:sample_size_table}}
\end{table}

\section{Real data application}

In this section we illustrate KERE by applying it to the Personal
Computer Price Data studied in \citet{stengos2006intertemporal}.
The data collected from the PC Magazine from January of 1993 to November
of 1995 has 6259 observations, each of which consists of the advertised
price and features of personal computers sold in United States. There
are 9 main price detriments of PCs summarized in Table \ref{tab:housing_data}.
The price and the continuous variables except the time trend are in
logarithmic scale. We consider a hedonic analysis, where the price
of a product is considered to be a function of the implicit prices
of its various components, see \citet{triplett1989price}. The intertemporal effect of the implicit PC-component prices is captured by
incorporating the time trend as one of the explanatory variables.
The presence of non-linearity and the interactions of the  components
with the time trend in the data, shown by \citet{stengos2006intertemporal},
suggest that the linear expectile regression
may lead to a misspecified model. Since there lacks of a general theory
about any particular functional form for the PC prices, we use KERE to capture the nonlinear effects and higher order interactions of characteristics on price
and avoid severe model misspecification.

We randomly sampled $1/10$ observations for training and tuning with
two-dimensional five-fold cross-validation for selecting an optimal
$(\sigma^{2},\lambda)$ pair, and the remaining  $9/10$ observations as
the test set for calculating the prediction error defined by
\[
\mathrm{prediction\ error}=\frac{1}{n^{\prime}}\sum_{i=1}^{n^{\prime}}\phi_{\omega}(y_{i}-\hat{f}_{\omega}(\mathbf{x}_{i})).
\]
For comparison, we also computed the prediction errors using the linear
expectile regression models under the same
setting. All prediction errors are computed for seven expectile levels
$\omega\in\{0.05,$ 0.1, 0.25, 0.5, 0.75, 0.9, $0.95\}$. We repeated
this process 100 times and reported the average prediction error and
their corresponding standard errors in Table \ref{tab:real}. We also
showed box-plots of empirical distributions of prediction errors in
Figure \ref{fig:real}. We see that for all expectile levels KERE
outperforms the linear expectile model in terms of both prediction
error and the corresponding standard errors. This shows that KERE
offers much more flexible and accurate predictions than the linear
model by guarding against model misspecification bias.

\begin{table}
\ra{1.1}

\begin{centering}
\begin{tabular}{lll}
\toprule
{\small{}{ID } } & {\small{}{Variable } } & {\small{}{Explanation}}\tabularnewline
\midrule
{\small{}{1 } } & SPEED & clock speed in MHz\tabularnewline
{\small{}{2 } } & HD & size of hard drive in MB\tabularnewline
{\small{}{3 } } & RAM & size of RAM in in MB\tabularnewline
{\small{}{4 } } & SCREEN & size of screen in inches\tabularnewline
{\small{}{5 } } & CD & if a CD-ROM present\tabularnewline
{\small{}{6 } } & PREMIUM & if the manufacturer was a ``premium'' firm (IBM, COMPAQ)\tabularnewline
{\small{}{7 } } & MULTI & if a multimedia kit (speakers, sound card) included\tabularnewline
{\small{}{8 } } & ADS & number of 486 price listings for each month\tabularnewline
{\small{}{9 } } & TREND & time trend indicating month starting from Jan. 1993 to Nov. 1995\tabularnewline
\bottomrule
\end{tabular}
\par\end{centering}

\caption{Explanatory variables in the Personal Computer Price Data \citep{stengos2006intertemporal} \label{tab:housing_data}}

\end{table}

\begin{table}
\ra{1.1}

\begin{centering}
\begin{tabular}{rrrrrrrr}
\toprule
\multicolumn{8}{c}{Personal Computer Price Data }\tabularnewline
\midrule
$\omega$ & 0.05 & 0.1 & 0.25 & 0.5 & 0.75 & 0.9 & 0.95\tabularnewline
\midrule
Linear & 5.727 & 3.396 & 5.722 & 7.078 & 6.032 & 3.814 & 2.517\tabularnewline
 & (0.013) & (0.010) & (0.015) & (0.017) & (0.015) & (0.014) & (0.012)\tabularnewline
KERE & 3.970 & 2.523 & 3.952 & 4.749 & 4.094 & 2.684 & 1.868\tabularnewline
 & (0.013) & (0.010) & (0.015) & (0.017) & (0.015) & (0.014) & (0.012)\tabularnewline
\bottomrule
\end{tabular}
\par\end{centering}

\caption{The averaged prediction error and the corresponding standard errors
for the Personal Computer Price Data based on 100 independent runs.
The expectile levels are $\omega\in\{0.05,0.1,0.25,0.5,0.75,0.9,0.95\}$. The numbers in this table are of the order of $10^{-3}.$\label{tab:real}}
\end{table}

\begin{figure}
\begin{centering}
\includegraphics[scale=0.65]{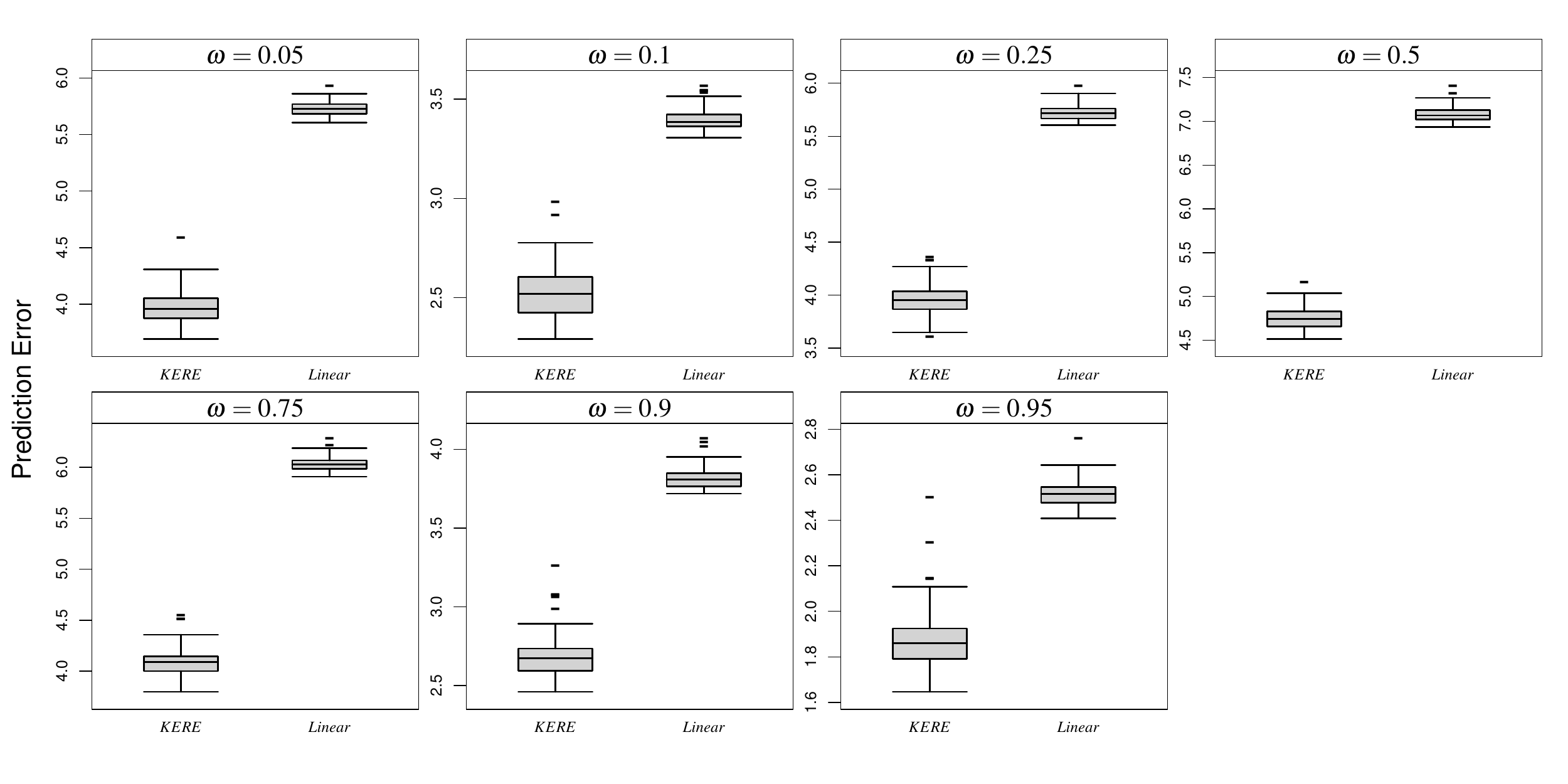}
\par\end{centering}

\caption{Prediction error distributions for the Personal Computer Price Data
using the linear expectile model and KERE. Box-plots show prediction error
based on 100 independent runs for expectiles $\omega\in\{0.05,0.1,0.25,0.5,0.75,0.9,0.95\}$. The numbers in this table are of the order of $10^{-3}.$
\label{fig:real}}
\end{figure}

\section*{Appendix: Technical Proofs}

\subsection{Some technical lemmas for Theorem~\ref{thm:asymptotic}}

We first present some technical lemmas and their proofs. These lemmas are used to prove Theorem~\ref{thm:asymptotic}.

\begin{lem}\label{lemma:setup3}
Let $\phi_{\omega}^*$ be the convex conjugate of $\phi_{\omega}$,
\[
\phi_{\omega}^*(t)=\begin{cases}
\frac{1}{4(1-\omega)}t^2 &\text{if $t \leq 0$,}\\
\frac{1}{4\omega}t^2&\text{if $t > 0$.}
\end{cases}
\]
The solution to \eqref{eq:setup2} can be alternatively obtained by solving the optimization problem
\begin{equation}
\min_{{\{\alpha_{i}\}_{i=0}^{n}}} g(\alpha_1,\alpha_2,\ldots, \alpha_n), \quad \text{subject to} \quad \sum_{i=1}^n\alpha_i=0,\label{eq:setup3}
\end{equation}
where $g$ is defined by
\begin{equation}
g(\alpha_1,\alpha_2,\ldots, \alpha_n)=-\sum_{i=1}^n y_i \alpha_i  +\frac{1}{2}\sum_{i,j=1}^n\alpha_{i}\alpha_{j}K(\bx_{i},\bx_{j}) +  2 \lambda \sum_{i=1}^n\phi_{\omega}^*(\alpha_i).\label{eq:gdef}
\end{equation}

\end{lem}
\begin{proof}
Let $\boldsymbol{\alpha}=(\alpha_1,\alpha_2,\ldots,\alpha_n)^{\intercal}$.  Since both objective functions in \eqref{eq:setup2} and \eqref{eq:setup3} are convex, we only need to show that they share a common stationary point. Define
\[
G_{\omega}(\boldsymbol{\alpha}) = \phi_{\omega}(\alpha_1)+\phi_{\omega}(\alpha_2)+\cdots+\phi_{\omega}(\alpha_n),
\]
\[
\nabla G_{\omega}(\boldsymbol{\alpha}) = (\phi^{\prime}_{\omega}(\alpha_1),\phi^{\prime}_{\omega}(\alpha_2),\ldots,\phi^{\prime}_{\omega}(\alpha_n))^{\intercal}.
\]
By setting the derivatives of \eqref{eq:setup2} with respect to $\boldsymbol{\alpha}$  to be zero, we can find the stationary point of \eqref{eq:setup2} satisfying
\[
\frac{\mathrm{d}}{\mathrm{d\boldsymbol{\alpha}}}\left[\left(\begin{array}{c}
y_{1}-\alpha_{0}\\
y_{2}-\alpha_{0}\\
\vdots\\
y_{n}-\alpha_{0}
\end{array}\right)-\mathbf{K}\boldsymbol{\alpha}\right]\cdot\left[\begin{array}{c}
\phi_{\omega}^{\prime}\big(y_{1}-\alpha_{0}-\sum_{j=1}^{n}K(x_{1},x_{j})\alpha_{j}\big)\\
\phi_{\omega}^{\prime}\big(y_{2}-\alpha_{0}-\sum_{j=1}^{n}K(x_{2},x_{j})\alpha_{j}\big)\\
\vdots\\
\phi_{\omega}^{\prime}\big(y_{n}-\alpha_{0}-\sum_{j=1}^{n}K(x_{n},x_{j})\alpha_{j}\big)
\end{array}\right]+\lambda\frac{\mathrm{d}}{\mathrm{d\boldsymbol{\alpha}}}\boldsymbol{\alpha}^{\intercal}K\boldsymbol{\alpha}=\mathbf{0},
\]
which can be reduced to
\begin{equation}\label{eq:setup2_1}
-\phi_{\omega}^{\prime}\big(y_i-\alpha_0-\sum_{j=1}^n K(x_i,x_j)\alpha_j\big)+2\lambda\alpha_i=0,\quad\text{for $1\leq i \leq n$},
\end{equation}
and setting the derivative of \eqref{eq:setup2} with respect to $\alpha_0$ to be zero, we have
\begin{equation}\label{eq:setup2_2}
\sum_{i=1}^n \phi_{\omega}'\big(y_i-\alpha_0-\sum_{j=1}^n K(x_i,x_j)\alpha_j\big) =0.
\end{equation}
Combining \eqref{eq:setup2_1} and \eqref{eq:setup2_2}, \eqref{eq:setup2_2} can be simplified to
\begin{equation}\label{eq:setup2_3}
\sum_{i=1}^n \alpha_i =0.
\end{equation}
In comparison, the Lagrange function of \eqref{eq:setup3} is
\begin{equation}
g(\alpha_1,\alpha_2,\ldots, \alpha_n)+\nu\sum_{i=1}^n\alpha_i.
\label{eq:lagrangeform}
\end{equation}
The first order conditions of \eqref{eq:lagrangeform} are
\begin{equation}\label{eq:setup3_1}
-y_i+\nu+\sum_{j=1}^n K(x_i,x_j)\alpha_j +2 \lambda \phi_{\omega}^{*\,'}(\alpha_i)=0,\quad \text{for $1\leq i \leq n$},
\end{equation}
and
\begin{equation}\label{eq:setup2_4}
\sum_{i=1}^n \alpha_i =0.
\end{equation}
Noting that $2 \lambda \phi_{\omega}^{*\,'}(\alpha_i)=\phi_{\omega}^{*\,'}(2 \lambda \alpha_i)$ and $\phi_{\omega}^{*\,'}$ is the inverse function of $\phi_{\omega}'$. Let $\nu=\alpha_0$, then  \eqref{eq:setup2_1} and \eqref{eq:setup3_1} are equivalent. Therefore, \eqref{eq:setup2} and \eqref{eq:setup3} have a common stationary point and therefore a common minimizer.
\end{proof}

\begin{lem}\label{lemma:norm_control}
\[
\sum_{j=1}^{n}\alpha_{j}K(\bx_{i},\bx_{j})\leq \sqrt{K(\bx_i,\bx_i)}\cdot \sqrt{\sum_{i=1}^{n}\sum_{j=1}^{n}\alpha_{i}\alpha_{j}K(\bx_{i},\bx_{j})}.
\]
\end{lem}
\begin{proof}
Let $\bC=\bK^{1/2}$, then by Cauchy-Schwarz inequality
\begin{align*}&\sum_{j=1}^{n}\alpha_{j}K(\bx_{i},\bx_{j})
=({\alpha}_1, {\alpha}_2, \ldots, {\alpha}_{n})\bC (\bC_{i,1},\bC_{i,2},\ldots, \bC_{i,n})^T
\\\leq &\|({\alpha}_1, {\alpha}_2, \ldots, {\alpha}_{n})\bC\|\cdot  \|(\bC_{i,1},\bC_{i,2},\ldots, \bC_{i,n})\|
= \sqrt{\sum_{i=1}^{n}\sum_{j=1}^{n}\alpha_{i}\alpha_{j}K(\bx_{i},\bx_{j})}\cdot \sqrt{K(\bx_i,\bx_i)}.
\end{align*}
\end{proof}

\begin{lem} \label{lemma:bound_objective}
For the $g$ function defined in \eqref{eq:gdef}, we have
\begin{align*}
&\frac{1}{2}\sum_{i,j=1}^{n}(\alpha_{i}-\hat{\alpha}_i)(\alpha_{j}-\hat{\alpha}_j)K(\bx_{i},\bx_{j})
+\frac{\lambda}{2\max(1-\omega,\omega)}\sum_{i=1}^{n}(\alpha_{i}-\hat{\alpha}_i)^2
\\\leq &
g(\alpha_1,\alpha_2,\ldots,\alpha_{n})
-
g(\hat{\alpha}_1,\hat{\alpha}_2,\ldots,\hat{\alpha}_{n})
\\\leq & \frac{1}{2}\sum_{i,j=1}^{n}(\alpha_{i}-\hat{\alpha}_i)(\alpha_{j}-\hat{\alpha}_j)K(\bx_{i},\bx_{j})
+\frac{\lambda}{2\min(1-\omega,\omega)}\sum_{i=1}^{n}(\alpha_{i}-\hat{\alpha}_i)^2.
\end{align*}
\end{lem}

\begin{proof}
It is clear that the second derivative of $g$ is bounded above by $\bK+\frac{\lambda}{\min(1-\omega,\omega)}\bI$ and bounded below by $\bK+\frac{\lambda}{\max(1-\omega,\omega)}\bI$, where $\bK\in\mathbb{R}^{n,n}$. Let $\boldsymbol{\alpha}=(\alpha_1,\alpha_2,\ldots,\alpha_{n})^{\intercal}$
\begin{eqnarray}
g(\boldsymbol{\alpha})-g(\widehat{\boldsymbol{\alpha}})\leq g^{\prime}(\widehat{\boldsymbol{\alpha}})^{\intercal}(\boldsymbol{\alpha}-\widehat{\boldsymbol{\alpha}})+\frac{1}{2}(\bK+\frac{\lambda}{\min(1-\omega,\omega)}\bI)(\boldsymbol{\alpha}-\widehat{\boldsymbol{\alpha}})^{\intercal}(\boldsymbol{\alpha}-\widehat{\boldsymbol{\alpha}}),\\
g(\boldsymbol{\alpha})-g(\widehat{\boldsymbol{\alpha}})\geq g^{\prime}(\widehat{\boldsymbol{\alpha}})^{\intercal}(\boldsymbol{\alpha}-\widehat{\boldsymbol{\alpha}})+\frac{1}{2}(\bK+\frac{\lambda}{\max(1-\omega,\omega)}\bI)(\boldsymbol{\alpha}-\widehat{\boldsymbol{\alpha}})^{\intercal}(\boldsymbol{\alpha}-\widehat{\boldsymbol{\alpha}}).
\end{eqnarray}
Hence when $\boldsymbol{\alpha}$ and $\widehat{\boldsymbol{\alpha}}$ are fixed and $g'(\widehat{\boldsymbol{\alpha}})=0$, the maximum of $g(\boldsymbol{\alpha})-g(\widehat{\boldsymbol{\alpha}})$ is obtained when the second order derivative of $g$ achieves its maximum and the minimum is obtained when the second order derivative achieves its minimum.
\end{proof}

The next lemma establishes the basis for the so-called leave-one-out analysis \citep{jaakkola1999probabilistic,joachims2000estimating,forster2002relative, ZhangTong2003}.
The basic idea is that the expected observed risk is equivalent to the expected leave-one-out error. Let $D_{n+1}=\{(\bx_{i},y_{i})\}_{i=1}^{n+1}$ be a random sample of size $n+1$, and let $D^{[i]}_{n+1}$ be the subset of $D_{n+1}$ with the $i$-th observation removed, i.e.
\[
D^{[i]}_{n+1}=\{(\mathbf{x}_1,y_1),\ldots,(\mathbf{x}_{i-1},y_{i-1}),(\mathbf{x}_{i+1},y_{i+1}),\ldots,(\mathbf{x}_{n+1},y_{n+1})\}.
\] 
Let $(\hat{f}^{[i]},\hat{\alpha}_0^{[i]})$ be the estimator trained on $D^{[i]}_{n+1}$. The leave-one-out error is defined as the averaged prediction error on each observation $(\mathbf{x}_i,y_i)$  using the estimator $(\hat{f}^{[i]},\hat{\alpha}_0^{[i]})$ computed from $D^{[i]}_{n+1}$, where $(\mathbf{x}_i,y_i)$ is excluded:
\[
\text{Leave-one-out error:}\quad\frac{1}{n+1}\sum^{n+1}_{i=1}\phi_{\omega}(y_i-\hat{\alpha}_{0}^{[i]}-\hat{f}^{[i]}(\bx_i)).
\]

\begin{lem}\label{lemma:leaveone}
Let $(\hat{f}_{(n)},\hat{\alpha}_{0\,(n)})$ be the KERE estimator trained from $D_{n}$. The expected observed risk $E_{D_{n}}E_{(\bx,y)}\phi_{\omega}(y-\hat{\alpha}_{0\,(n)}-\hat{f}_{(n)}(\bx))$
is equivalent to the expected leave-one-out error on $D_{n+1}$:
\begin{equation}
E_{D_{n}}\big\{E_{(\bx,y)}\phi_{\omega}(y-\hat{\alpha}_{0\,(n)}-\hat{f}_{(n)}(\bx))\big\} = E_{D_{n+1}}\Big(\frac{1}{n+1}\sum^{n+1}_{i=1}\phi_{\omega}(y_i-\hat{\alpha}_{0}^{[i]}-\hat{f}^{[i]}(\bx_i))\Big),\label{eq:leave_one_out}
\end{equation}
where $\hat{\alpha}_{0}^{[i]}$ and $\hat{f}^{[i]}$ are KERE trained from $D^{[i]}_{n+1}$.
\end{lem}
\begin{proof}
\begin{eqnarray*}
E_{D_{n+1}}\Big(\frac{1}{n+1}\sum_{i=1}^{n+1}\phi_{\omega}(y_{i}-\hat{\alpha}_{0}^{[i]}-\hat{f}^{[i]}(\bx_{i}))\Big) & = & \frac{1}{n+1}\sum_{i=1}^{n+1}E_{D_{n+1}}\phi_{\omega}(y_{i}-\hat{\alpha}_{0}^{[i]}-\hat{f}^{[i]}(\bx_{i}))\\
 & = & \frac{1}{n+1}\sum_{i=1}^{n+1}E_{D_{n+1}^{[i]}}\big\{E_{(\mathbf{x}_{i},y_{i})}\phi_{\omega}(y_{i}-\hat{\alpha}_{0}^{[i]}-\hat{f}^{[i]}(\bx_{i}))\big\}\\
 & = & \frac{1}{n+1}\sum_{i=1}^{n+1}E_{D_{n}}\big\{E_{(\bx,y)}\phi_{\omega}(y-\hat{\alpha}_{0}-\hat{f}(\bx))\big\}\\
 & = & E_{D_{n}}{E_{(\bx,y)}\phi_{\omega}(y-\hat{\alpha}_{0}-\hat{f}(\bx))}.
\end{eqnarray*}
\end{proof}

In the following Lemma, we give an upper bound of $|\hat{\alpha}_{i}|$ for $1\leq i\leq n$.
\begin{lem}\label{lemma:alphaibound}
Assume $M=\sup_{\bx}K(\bx,\bx)^{1/2}$.
Denote  as $(\hat{f}_{(n)},\hat{\alpha}_{0\,(n)})$ the KERE estimator in \eqref{eq:setup1}  trained on $n$ samples $D_{n}=\{(\bx_{i},y_{i})\}_{i=1}^{n}$.
The estimates $\hat{\alpha}_{i\,(n)}$ for $1\leq i\leq n$ are defined by $\hat{f}_{(n)}(\cdot)=\sum_{i=1}^{n}\hat{\alpha}_{i\,(n)}K(\bx_{i},\cdot)$.  Denote $\|Y_n\|_{2}=\sqrt{\sum_{i=1}^{n}y_{i}^{2}}$, $\frac{\|Y_{n}\|_1}{n}=\frac{1}{n}\sum_{i=1}^{n}|y_{i}|$, $q_{1}=\frac{\max(1-\omega,\omega)}{\min(1-\omega,\omega)}$, $q_{2}=\max(1-\omega,\omega)$. We claim that
\begin{align}\label{eq:est_alphai}
|\hat{\alpha}_{i\,(n)}|& \leq  \frac{ q_{2} }{\lambda}\Big(q_{1}  \frac{\|Y_{n}\|_1}{n}+ M(q_{1}+1)\sqrt{\frac{ q_{2}}{\lambda}}\|Y_n\|_{2}+|y_{i}|\Big),\quad \text{for\ } 1\leq i\leq n.
\end{align}
\end{lem}
\begin{proof}

The proof is as follows. The function $g$ is defined as in \eqref{eq:gdef}, then
\[
g(\hat{\alpha}_{1\,(n)}, \hat{\alpha}_{2\,(n)}, \ldots, \hat{\alpha}_{n\,(n)})
\leq g(0, 0, \ldots, 0) = 0,
\]
we have
\begin{eqnarray*}
\frac{1}{2}\sum_{i,j=1}^{n}\hat{\alpha}_{i\,(n)}\hat{\alpha}_{j\,(n)}K(\bx_{i},\bx_{j}) & \leq & \sum_{i=1}^{n}y_{i}\hat{\alpha}_{i\,(n)}-2\lambda\sum_{i=1}^{n}\phi_{\omega}^{*}(\hat{\alpha}_{i\,(n)})\\
 & \leq & -\frac{\lambda}{2 q_{2} }\sum_{i=1}^{n}\Big(\hat{\alpha}_{i\,(n)}-\frac{ q_{2} }{\lambda}y_i\Big)^2 + \frac{ q_{2} }{2\lambda}\sum_{i=1}^{n}y_{i}^{2}\\
  & \leq & \frac{ q_{2} }{2\lambda}\sum_{i=1}^{n}y_{i}^{2}.
\end{eqnarray*}
Applying Lemma~\ref{lemma:norm_control}, we have
\begin{equation}\label{eq:est_Kalpha}
\hat{f}_{(n)}(\bx_i)=\sum_{j=1}^{n}\hat{\alpha}_{j\,(n)}K(\bx_{i},\bx_{j})
\leq  M\sqrt{\frac{ q_{2} \,\sum_{i=1}^{n}y_{i}^{2}}{\lambda}}
= M \sqrt{\frac{ q_{2}}{\lambda}}\|Y_n\|_{2}.\end{equation}
By the definition in \eqref{eq:setup2}, $\hat{\alpha}_{0\,(n)}$ is given by
$
\argmin_{\alpha_0}\sum_{i=1}^{n}\phi_{\omega}\big(y_{i}-\alpha_{0}-\hat{f}_{(n)}(\bx_i)\big).$
By the first order condition
\[
\sum_{i=1}^{n}2\big|\omega-I(y_i-\hat{\alpha}_{0\,(n)}-\hat{f}_{(n)}(\mathbf{x}_i))\big|(y_i-\hat{\alpha}_{0\,(n)}-\hat{f}_{(n)}(\mathbf{x}_i))=0.
\]
Let $c_i = \big|\omega-I(y_i-\hat{\alpha}_{0\,(n)}-\hat{f}_{(n)}(\mathbf{x}_i))\big|$, we have $\min(1-\omega,\omega)\leq c_i\leq\max(1-\omega,\omega)$, hence
\begin{eqnarray*}
\Big|\Big(\sum_{i=1}^{n}c_{i}\Big)\hat{\alpha}_{0\,(n)}\Big| & = & \Big|\sum_{i=1}^{n}c_{i}(y_{i}-\hat{f}_{(n)}(\mathbf{x}_{i}))\Big|\leq\sum_{i=1}^{n}c_{i}(\big|y_{i}\big|+\big|\hat{f}_{(n)}(\mathbf{x}_{i})\big|)\\
 & \leq &  q_{2}  \Big(\sum_{i=1}^n |y_i|+nM \sqrt{\frac{ q_{2}}{\lambda}}\|Y_n\|_{2}\Big),
\end{eqnarray*}
and  we have
\begin{equation}\label{eq:est_alpha0}
|\hat{\alpha}_{0\,(n)}|\leq   q_{1}  \Big(\frac{\|Y_{n}\|_1}{n}+ M \sqrt{\frac{ q_{2}}{\lambda}}\|Y_n\|_{2}\Big).
\end{equation}
Combining \eqref{eq:setup2_1} and \eqref{eq:est_alpha0}, we concluded \eqref{eq:est_alphai}.
\end{proof}

\subsection{Proof of Theorem~\ref{thm:asymptotic}}

\begin{proof}  Consider $n+1$ training samples $D_{n+1}=\{(\mathbf{x}_1,y_1),\ldots,(\mathbf{x}_{n+1},y_{n+1})\}$. Denote as $(\hat{f}^{[i]},\hat{\alpha}_0^{[i]})$ the KERE estimator trained from $D^{[i]}_{n+1}$, which is a subset of $D_{n+1}$ with $i$-th observation removed, i.e., 
\[
D^{[i]}_{n+1}=\{(\mathbf{x}_1,y_1),\ldots,(\mathbf{x}_{i-1},y_{i-1}),(\mathbf{x}_{i+1},y_{i+1}),\ldots,(\mathbf{x}_{n+1},y_{n+1})\}.
\]
Denote  as $(\hat{f}_{(n+1)},\hat{\alpha}_{0\,(n+1)})$ the KERE estimator trained from $n+1$ samples $D_{n+1}$. The estimates  $\hat{\alpha}_{i}$ for $1\leq i\leq n+1$ are defined by $\hat{f}_{(n+1)}(\cdot)=\sum_{i=1}^{n+1}\hat{\alpha}_{i}K(\bx_{i},\cdot)$.

In what follows,  we denote $\|Y_{n+1}\|_{2}=\sqrt{\sum_{i=1}^{n+1}y_{i}^{2}}$, $\frac{\|Y_{n+1}\|_1}{n+1}=\frac{1}{n+1}\sum_{i=1}^{n+1}|y_{i}|$, $q_{1}=\frac{\max(1-\omega,\omega)}{\min(1-\omega,\omega)}$, $q_{2}=\max(1-\omega,\omega)$, $q_{3}=\min(1-\omega,\omega)$. 

\paragraph{Part I} 
We first show that the leave-one-out estimate is sufficiently close to the estimate fitted from using all the training data. 
Without loss of generality, just consider the case that the $(n+1)$th data point is removed. The same results apply to the other leave-one out cases.
We show that
$
|\hat{f}^{[n+1]}(\bx_{i})+\hat{\alpha}_0^{[n+1]}-\hat{f}_{(n+1)}(\bx_{i})-\hat{\alpha}_{0\,(n+1)}|\leq  C^{[n+1]}_{2},
$
where the expression of 
$
C^{[n+1]}_{2} 
$
is to be derived in the following.

We first study the upper bound for $|\hat{f}^{[n+1]}(\bx_i)-\hat{f}_{(n+1)}(\bx_i)|$.
By the definitions of $g$ in \eqref{eq:gdef} and $(\hat{\alpha}_{1}^{[n+1]},\hat{\alpha}_{2}^{[n+1]},\ldots,\hat{\alpha}_{n}^{[n+1]})$, we have
\begin{eqnarray*}
 &  & g\big(\hat{\alpha}_{1}^{[n+1]},\hat{\alpha}_{2}^{[n+1]},\ldots,\hat{\alpha}_{n}^{[n+1]},0\big)\\
 & = & g\big(\hat{\alpha}_{1}^{[n+1]},\hat{\alpha}_{2}^{[n+1]},\ldots,\hat{\alpha}_{n}^{[n+1]}\big)\\
 & \leq & g\big(\hat{\alpha}_{1}+\frac{1}{n}\hat{\alpha}_{n+1},\hat{\alpha}_{2}+\frac{1}{n}\hat{\alpha}_{n+1},\ldots,\hat{\alpha}_{n}+\frac{1}{n}\hat{\alpha}_{n+1}\big)\\
 & = & g\big(\hat{\alpha}_{1}+\frac{1}{n}\hat{\alpha}_{n+1},\hat{\alpha}_{2}+\frac{1}{n}\hat{\alpha}_{n+1},\ldots,\hat{\alpha}_{n}+\frac{1}{n}\hat{\alpha}_{n+1},0\big).
\end{eqnarray*}
That is,
\begin{align*}
&g\big(\hat{\alpha}_{1}^{[n+1]},\hat{\alpha}_{2}^{[n+1]},\ldots,\hat{\alpha}_{n}^{[n+1]},0\big)
-g\big(\hat{\alpha}_{1},\hat{\alpha}_{2},\ldots,\hat{\alpha}_{n+1}\big)\\
&\leq g\big(\hat{\alpha}_{1}+\frac{1}{n}\hat{\alpha}_{n+1},\hat{\alpha}_{2}+\frac{1}{n}\hat{\alpha}_{n+1},\ldots,\hat{\alpha}_{n}+\frac{1}{n}\hat{\alpha}_{n+1},0\big)
-g\big(\hat{\alpha}_{1},\hat{\alpha}_{2},\ldots,\hat{\alpha}_{n+1}\big).
\end{align*}
Denote for simplicity that $\hat{\alpha}_{n+1}^{[n+1]}=0$. Applying Lemma \ref{lemma:bound_objective} to both LHS and RHS of the above inequality, we have
\begin{align*}
 & \sum_{i,j=1}^{n+1}(\hat{\alpha}_{i}^{[n+1]}-\hat{\alpha}_{i})(\hat{\alpha}_{j}^{[n+1]}-\hat{\alpha}_{j})K(\bx_{i},\bx_{j})
+\frac{\lambda}{2 q_{2} }\sum_{i=1}^{n+1}(\hat{\alpha}_i^{[n+1]}-\hat{\alpha}_i)^2\\
\leq  & \hat{\alpha}_{n+1}^{2}\Big[\Big(\frac{1}{n},\ldots,\frac{1}{n},-1\Big)\bK\Big(\frac{1}{n},\ldots,\frac{1}{n},-1\Big)^{T}+\frac{\lambda(n+1)}{2n q_{3} }\Big],
\end{align*}
where $\bK\in\mathbb{R}^{n+1,n+1}$ is defined by $\bK_{i,j}=K(\bx_i,\bx_j)$.
Since $|K(\bx_i,\bx_j)|\leq M^{2}$ for any $1\leq i,j\leq n+1$, we have
\[
\begin{array}{ll}
&\Big(\frac{1}{n},\ldots, \frac{1}{n}, -1 \Big) \bK \Big(\frac{1}{n},\ldots, \frac{1}{n}, -1 \Big)^T\\=&\frac{1}{n^2}\sum_{i,j=1}^n \bK_{i,j}-\frac{1}{n}\sum_{i=1}^n\bK_{i,n+1}
-\frac{1}{n}\sum_{j=1}^n\bK_{n+1,j}+\bK_{n+1,n+1}
\\\leq& M^{2}+M^{2}+M^{2}+M^{2} = 4M^{2}.
\end{array}
\]
Combining it with the bound for $|\hat{\alpha}_{n+1}|$ by Lemma \ref{lemma:alphaibound} (note that here $\hat{\alpha}_{n+1}$ is trained on $n+1$ samples), we have
\begin{equation}\label{eq:firstineq}
\sum_{i,j=1}^{n+1}(\hat{\alpha}_i^{[n+1]}-\hat{\alpha}_i)(\hat{\alpha}_j^{[n+1]}-\hat{\alpha}_j)K(\bx_{i},\bx_{j})
\leq C^{[n+1]}_{1},
\end{equation}
where
\begin{align}
C^{[n+1]}_{1} &= \Bigg(4M^{2}+ \frac{\lambda(n+1)}{2n q_{3} }\Bigg)\Bigg(\frac{ q_{2} }{\lambda} C^{[n+1]}_{0}  \Bigg)^2, \label{eq:c0}
\end{align}
and
\begin{align}\label{eq:truec0}
C^{[n+1]}_{0} &= q_{1}  \frac{\|Y_{n+1}\|_1}{n+1}+ M(q_{1}+1)\sqrt{\frac{ q_{2}}{\lambda}}\|Y_{n+1}\|_{2}+|y_{n+1}|.
\end{align}
Combining \eqref{eq:firstineq} with Lemma~\ref{lemma:norm_control}, we have that for $1\leq i\leq n+1$,
\begin{equation}\label{eq:est_diff_alphaK}
|\hat{f}^{[n+1]}(\bx_i)-\hat{f}_{(n+1)}(\bx_i)|=
\Big|\sum_{j=1}^{n+1}(\hat{\alpha}_i^{[n+1]}-\hat{\alpha}_i)K(\bx_{i},\bx_{j})\Big|\leq \sqrt{C^{[n+1]}_{1}}M.
\end{equation}

Next, we bound $|\hat{\alpha}_0^{[n+1]}-\hat{\alpha}_{0\,(n+1)}|$. Since $\hat{\alpha}_{0\,(n+1)}$ and $\hat{\alpha}_0^{[n+1]}$ are the minimizers of
\[\text{$\sum_{i=1}^{n+1}\phi_{\omega}\left(y_{i}-\alpha_{0}-\hat{f}_{(n+1)}(\bx_i)\right)$
and $\sum_{i=1}^{n}\phi_{\omega}\left(y_{i}-\alpha_{0}-\hat{f}^{[n+1]}(\bx_i)\right)$},\]
we have
\begin{equation}\label{eq:a0firstorder}
\frac{\di }{\di \alpha_0}\sum_{i=1}^{n+1}\phi_{\omega}\left(y_{i}-\alpha_{0}-\hat{f}_{(n+1)}(\bx_i)\right)\Big|_{\alpha_0=\hat{\alpha}_{0\,(n+1)}}=0,
\end{equation}
and
\begin{equation}\label{eq:a0firstorder2}
\frac{\di }{\di \alpha_0}\sum_{i=1}^{n}\phi_{\omega}\left(y_{i}-\alpha_{0}-\hat{f}^{[n+1]}(\bx_i)\right)\Big|_{\alpha_0=\hat{\alpha}_0^{[n+1]}}
=0.
\end{equation}
By the Lipschitz  continuity of $\phi'_{\omega}$ we have
\[
\begin{array}{ll}
 & \Bigg|\sum_{i=1}^{n+1}\phi'_{\omega}\left(y_{i}-\hat{\alpha}_{0\,(n+1)}-\hat{f}^{[n+1]}(\bx_{i})\right)-\sum_{i=1}^{n+1}\phi'_{\omega}\left(y_{i}-\hat{\alpha}_{0\,(n+1)}-\hat{f}_{(n+1)}(\bx_{i})\right)\Bigg|\\
\leq & 2(n+1) q_{2} |\hat{f}^{[n+1]}(\bx_{i})-\hat{f}_{(n+1)}(\bx_{i})|,
\end{array}
\]
and by applying \eqref{eq:est_diff_alphaK} and \eqref{eq:a0firstorder} we have the upper bound
\[
\Bigg|\sum_{i=1}^{n+1}\phi'_{\omega}\left(y_{i}-\hat{\alpha}_{0\,(n+1)}-\hat{f}^{[n+1]}(\bx_{i})\right)\Bigg|
\leq {2(n+1) q_{2} }\, \sqrt{C^{[n+1]}_{1}}M.
\]
Similarly, by \eqref{eq:est_Kalpha}, \eqref{eq:est_alpha0}, and \eqref{eq:a0firstorder2} we have
\begin{equation}
\begin{array}{ll}
 & \Bigg|\sum_{i=1}^{n}\phi'_{\omega}\left(y_{i}-\hat{\alpha}_{0\,(n+1)}-\hat{f}^{[n+1]}(\bx_{i})\right)\Bigg|\\
  = & \Bigg|\sum_{i=1}^{n+1}\phi'_{\omega}\left(y_{i}-\hat{\alpha}_{0\,(n+1)}-\hat{f}^{[n+1]}(\bx_{i})\right)-\sum_{i=1}^{n+1}\phi'_{\omega}\left(y_{i}-\hat{\alpha}_{0\,(n+1)}-\hat{f}_{(n+1)}(\bx_{i})\right)\\
   & -\phi'_{\omega}\left(y_{n+1}-\hat{\alpha}_{0\,(n+1)}-\hat{f}^{[n+1]}(\bx_{n+1})\right)\Bigg|\\
  \leq & \Bigg|\sum_{i=1}^{n+1}\phi'_{\omega}\left(y_{i}-\hat{\alpha}_{0\,(n+1)}-\hat{f}^{[n+1]}(\bx_{i})\right)-\sum_{i=1}^{n+1}\phi'_{\omega}\left(y_{i}-\hat{\alpha}_{0\,(n+1)}-\hat{f}_{(n+1)}(\bx_{i})\right)\Bigg|\\
   & +\Bigg|\phi'_{\omega}\left(y_{n+1}-\hat{\alpha}_{0\,(n+1)}-\hat{f}^{[n+1]}(\bx_{n+1})\right)\Bigg|\\
  \leq & 2(n+1) q_{2} \sqrt{C^{[n+1]}_{1}}M+2 q_{2} \big(|y_{n+1}|+ |\hat{\alpha}_{0\,(n+1)}| + |\hat{f}_{(n)}| \big)\\
  \leq & 2(n+1) q_{2} \sqrt{C^{[n+1]}_{1}}M+2 q_{2} \big(|y_{n+1}|+q_{1}  \frac{\|Y_{n+1}\|_1}{n+1}+ Mq_{1}\sqrt{\frac{ q_{2}}{\lambda}}\|Y_{n+1}\|_{2}+\sqrt{\frac{ q_{2}}{\lambda}}\|Y_{n}\|_{2}\big)\\
    \leq & 2(n+1) q_{2} \sqrt{C^{[n+1]}_{1}}M+2 q_{2} C^{[n+1]}_{0},
\end{array}\label{eq:interbound}
\end{equation}
where the second last inequality follows from \eqref{eq:est_Kalpha} and \eqref{eq:est_alpha0}. Note that in this case the corresponding sample is $n+1$.

Using \eqref{eq:a0firstorder2} we have
\begin{eqnarray*}
 &  & 2n q_{3} \big|\hat{\alpha}_{0}^{[n+1]}-\hat{\alpha}_{0\,(n+1)}\big|\\
 & \leq & \Big|\sum_{i=1}^{n}\phi_{\omega}^{\prime}\left(y_{i}-\hat{\alpha}_{0}^{[n+1]}-\hat{f}^{[n+1]}(\bx_{i})\right)-\sum_{i=1}^{n}\phi_{\omega}^{\prime}\left(y_{i}-\hat{\alpha}_{0\,(n+1)}-\hat{f}^{[n+1]}(\bx_{i})\right)\Big|\\
 & = & \Big|\sum_{i=1}^{n}\phi_{\omega}^{\prime}\left(y_{i}-\hat{\alpha}_{0\,(n+1)}-\hat{f}^{[n+1]}(\bx_{i})\right)\Big|.
\end{eqnarray*}
By \eqref{eq:interbound}, we have
\begin{equation}\label{eq:est_diff_alpha0}
|\hat{\alpha}_{0}^{[n+1]}-\hat{\alpha}_{0\,(n+1)}|\leq q_{1} \Big((1+\frac{1}{n})\sqrt{C^{[n+1]}_{1}}M+\frac{1}{n}C^{[n+1]}_{0}\Big).
\end{equation}
Finally, combining \eqref{eq:est_diff_alphaK} and \eqref{eq:est_diff_alpha0} we have
\begin{equation}\label{eq:diff_function}
|\hat{f}^{[n+1]}(\bx_{i})+\hat{\alpha}_0^{[n+1]}-\hat{f}_{(n+1)}(\bx_{i})-\hat{\alpha}_{0\,(n+1)}|\leq  C^{[n+1]}_{2},
\end{equation}
where
\begin{equation}\label{eq:c222}
C^{[n+1]}_{2} =  q_{1} \Big((1+\frac{1}{n})\sqrt{C^{[n+1]}_{1}}M+\frac{1}{n}C^{[n+1]}_{0}\Big)+ \sqrt{C^{[n+1]}_{1}}M.  
\end{equation}

\paragraph{Part II}  We now use \eqref{eq:diff_function} to derive a bound for $\phi_{\omega}(y_{n+1}-\hat{\alpha}_{0}^{[n+1]}-\hat{f}^{[n+1]}(\bx_{n+1}))$. 
Let $t = \hat{f}_{(n+1)}(\bx_{i})+\hat{\alpha}_{0\,(n+1)}-\hat{f}^{[n+1]}(\bx_{i})-\hat{\alpha}_0^{[n+1]}$ and $t^{\prime}=y_{i}-\hat{\alpha}_{0\,(n+1)}-\hat{f}_{(n+1)}(\bx_{i})$. We claim that, 
\begin{equation}\phi_\omega(t+t^{\prime})-\phi_\omega(t^{\prime})\leq  q_{2} (|2tt^{\prime}|+|t^2|).\label{eq:phi_diff}
\end{equation} 
when $(t+t^{\prime})$ and $t^{\prime}$ are both positive or both negative, \eqref{eq:phi_diff} follows from $(t+t^{\prime})^2-t^{\prime 2}=2tt^{\prime}+t^2$. When $t+t^{\prime}$ and $t^{\prime}$ have different signs, it must be that $|t^{\prime}|<|t|$, and we have $|t|=|t+t^{\prime}|+|t^{\prime}|$ and hence $|t+t^{\prime}|<|t|$. Then \eqref{eq:phi_diff} is proved by $\phi_\omega(t+t^{\prime})-\phi_\omega(t^{\prime})=\max(\phi_\omega(t+t^{\prime}),\phi_\omega(t^{\prime}))\leq
 q_{2} \max((t+t^{\prime})^2,t^{\prime 2})\leq \max(1-\omega,\omega) t^2 < \max(1-\omega,\omega)(|2tt^{\prime}|+|t^2|)$.

Hence by \eqref{eq:diff_function}, \eqref{eq:phi_diff} and the upper bound of $|y_{n+1}-\hat{f}_{(n+1)}(\bx_{n+1})-\hat{\alpha}_{0\,(n+1)}|$, we have 
\begin{align}\label{eq:c22}
&\phi_{\omega}(y_{n+1}-\hat{\alpha}_{0}^{[n+1]}-\hat{f}^{[n+1]}(\bx_{n+1}))
\leq
\phi_{\omega}(y_{n+1}-\hat{\alpha}_{0\,(n+1)}-\hat{f}_{(n+1)}(\bx_{n+1}))+C^{[n+1]}_{3},
\end{align}
where
\begin{align}
C^{[n+1]}_{3}&={ q_{2} } (2C^{[n+1]}_{0}C^{[n+1]}_{2}+(C^{[n+1]}_{2})^{2}).\label{eq:c2}
\end{align}
Note that \eqref{eq:c22} and \eqref{eq:c2} hold for other $i,  1 \le i \le n+1$.
\begin{align}
&\phi_{\omega}(y_{i}-\hat{\alpha}_{0}^{[i]}-\hat{f}^{[i]}(\bx_{i}))
\leq
\phi_{\omega}(y_{i}-\hat{\alpha}_{0\,(n+1)}-\hat{f}_{(n+1)}(\bx_{i}))+C^{[i]}_{3}.
\end{align}
Hence by \eqref{eq:c0}, \eqref{eq:truec0}, \eqref{eq:c222} and \eqref{eq:c22} we have
\begin{eqnarray}\label{eq:c2results}
E_{D_{n+1}}\Big(\phi_{\omega}(y_{i}-\hat{\alpha}_{0}^{[i]}-\hat{f}^{[i]}(\bx_{i}))\Big)\leq E_{D_{n+1}}\Big(\phi_{\omega}(y_{i}-\hat{\alpha}_{0\,(i)}-\hat{f}_{(n+1)}(\bx_{i}))\Big)+E_{D_{n+1}}C_{3}^{[i]}.
\end{eqnarray}
and
\begin{eqnarray}
&&\frac{1}{n+1}E_{D_{n+1}}\Big(\sum_{i=1}^{n+1}\phi_{\omega}(y_{i}-\hat{\alpha}_{0}^{[i]}-\hat{f}^{[i]}(\bx_{i}))\Big)\nonumber\\ 
& \leq & 
\frac{1}{n+1}E_{D_{n+1}}\Big(\sum_{i=1}^{n+1}\phi_{\omega}(y_{i}-\hat{\alpha}_{0\,(n+1)}-\hat{f}_{(n+1)}(\bx_{i}))\Big) +\frac{1}{n+1}E_{D_{n+1}}\sum_{i=1}^{n+1}C_{3}^{[i]}. \label{eq:ed1}
 \end{eqnarray}

On the other hand,  let $(f^*_\varepsilon, \alpha_{0\,\varepsilon}^*)$ in the RKHS and satisfy ${\cal R}(f^*_\varepsilon, \alpha_{0\,\varepsilon}^*)\leq \inf_{f\in\mathbb{H}_K,\alpha_0\in\mathbb{R}}{\cal R}(f, \alpha_{0})+\varepsilon$. From the definition of $\hat{\alpha}_{0\,(n+1)}$ and $\hat{f}_{(n+1)}$ we have
 \begin{eqnarray}
&&\frac{1}{n+1}\Big(\sum_{i=1}^{n+1}\phi_{\omega}(y_{i}-\hat{\alpha}_{0\,(n+1)}-\hat{f}_{(n+1)}(\bx_{i}))\Big)+\frac{\lambda}{n+1}\|\hat{f}_{(n+1)}\|_{\mathbb{H}_{K}}^{2} \nonumber\\
& \leq & \frac{1}{n+1}\Big(\sum_{i=1}^{n+1}\phi_{\omega}(y_{i}-\alpha_{0\,\varepsilon}^{*}-f_{\varepsilon}^{*}(\bx_{i}))\Big)
  +\frac{\lambda}{n+1}\|f_{\varepsilon}^{*}\|_{\mathbb{H}_{K}}^{2}. \label{eq:ed2}
\end{eqnarray}
By Lemma \ref{lemma:leaveone},  \eqref{eq:ed1} and \eqref{eq:ed2}, we get
\begin{eqnarray}
&&E_{D_{n}}\big\{ E_{(\bx,y)}\phi_{\omega}(y-\hat{\alpha}_{0\,(n)}-\hat{f}_{(n)}(\bx))\big\} \nonumber \\
& = & \frac{1}{n+1}E_{D_{n+1}}\Big(\sum_{i=1}^{n+1}\phi_{\omega}(y_{i}-\hat{\alpha}_{0}^{[i]}-\hat{f}^{[i]}(\bx_{i}))\Big)\nonumber \\
& \leq & E_{D_{n}}\big\{ E_{(\bx,y)}\phi_{\omega}(y-\alpha_{0\,\varepsilon}^{*}-f_{\varepsilon}^{*}(\bx_{i}))\big\}+\frac{\lambda}{n+1}\|f_{\varepsilon}^{*}\|_{\mathbb{H}_{K}}^{2}+\frac{1}{n+1}E_{D_{n+1}}\sum_{i=1}^{n+1}C_{3}^{[i]} \nonumber \\
& \le &  \inf_{f\in\mathbb{H}_K,\alpha_0\in\mathbb{R}}{\cal R}(f, \alpha_{0})+\varepsilon+\frac{\lambda}{n+1}\|f_{\varepsilon}^{*}\|_{\mathbb{H}_{K}}^{2}+\frac{1}{n+1}E_{D_{n+1}}\sum_{i=1}^{n+1}C_{3}^{[i]}.\label{eq:secondfinal}
\end{eqnarray}
Because $\lambda/n\rightarrow 0$, there exists $N_{\varepsilon}$ such that when $n>N_{\varepsilon}$, $\frac{\lambda}{n+1}\|f_{\varepsilon}^{*}\|_{\mathbb{H}_{K}}^{2} \le \varepsilon$.
In what follows, we show that there exists $N'_{\varepsilon}$ such that when $n>N'_{\varepsilon}$, $\frac{1}{n+1}E_{D_{n+1}}\sum_{i=1}^{n+1}C_{3}^{[i]} \le \varepsilon$. Thus, when $n>\max(N_{\varepsilon},N'_{\varepsilon})$ we have
$$
E_{D_{n}}\big\{ E_{(\bx,y)}\phi_{\omega}(y-\hat{\alpha}_{0\,(n)}-\hat{f}_{(n)}(\bx))\big\} \le \inf_{f\in\mathbb{H}_K,\alpha_0\in\mathbb{R}}{\cal R}(f, \alpha_{0})+3\varepsilon.
$$
Since it holds for any $\varepsilon>0$, Theorem~\ref{thm:asymptotic} will be proved.

Now we only need to show  that $\frac{1}{n+1}E_{D_{n+1}}\sum_{i=1}^{n+1}C_{3}^{[i]}\rightarrow 0$ as $n\rightarrow\infty$.
In fact we can show $\frac{1}{n+1}E_{D_{n+1}}\sum_{i=1}^{n+1}C_{3}^{[i]} \le \frac{C}{\sqrt{\lambda}}D\left(\frac{1+n}{\lambda}+1\right)\rightarrow 0$ as $n\rightarrow\infty$.
 In the following analysis, $C$ represents any constant that does not depend on $n$, but the value of $C$ varies in different expressions. 
Let $V_{i}=q_{1}  \frac{\|Y_{n+1}\|_1}{n+1}+ M(q_{1}+1)\sqrt{\frac{ q_{2}}{\lambda}}\|Y_{n+1}\|_{2}+|y_{i}|$, 
then  as $n\rightarrow\infty$, $4M^2<\frac{\lambda(n+1)}{2n q_{3} }$, and we have the upper bound
\[
C_{1}^{[i]}<(C\lambda)\Big(\frac{C}{\lambda}V_{i}\Big)^2=C\frac{V_{i}^{2}}{\lambda},
\]
and since $n>\sqrt{\lambda}$ asymptotically, we have
\[
C_{2}^{[i]}  <  C\Big(C\sqrt{C_{1}^{[i]}}+\frac{V_{i}}{n}\Big)+C\sqrt{C_{1}^{[i]}}<C\frac{V_{i}}{\sqrt{\lambda}}+C\frac{V_{i}}{n}<C\frac{V_{i}}{\sqrt{\lambda}}.
\]
Then 
\begin{equation}\label{eq:EC2} 
C_{3}^{[i]} < CV_{i}C_{2}^{[i]}+CC_{2}^{[i]\,2} < C V_{i} \frac{V_{i}}{\sqrt{\lambda}}+ C \frac{V_{i}^{2}}{{\lambda}}<C  \frac{V_{i}^{2}}{\sqrt{\lambda}}.
\end{equation}
We can bound $V_{i}$ as follows:
\begin{eqnarray*}
V_{i} & = & q_{1}  \frac{\|Y_{n+1}\|_1}{n+1}+ M(q_{1}+1)\sqrt{\frac{ q_{2}}{\lambda}}\|Y_{n+1}\|_{2}+|y_{i}|\\
 & < & q_{1}{\frac{\|Y_{n+1}\|_{2}}{\sqrt{n+1}}}+M(q_{1}+1)\sqrt{\frac{ q_{2}}{\lambda}}\|Y_{n+1}\|_{2}+|y_{i}|\\
 & < & C\sqrt{\frac{\|Y_{n+1}\|^2_{2}}{\lambda}}+C|y_{i}|.
\end{eqnarray*}
Then we have 
\begin{eqnarray}\label{eq:EC22} 
E_{D_{n+1}}V_{i}^{2} & < & 2C^2 E_{D_{n+1}}\Big[{\frac{\|Y_{n+1}\|^2_{2}}{\lambda}}+y^2_{i}\Big].
\end{eqnarray}
Combining it with \eqref{eq:EC2} and using the assumption $ E{y}^2_i<D$, we have
\begin{eqnarray*}
\frac{1}{n+1}E_{D_{n+1}}\sum_{i=1}^{n+1}C_{3}^{[i]} & \le & \frac{C}{\sqrt{\lambda}}\frac{1}{1+n}\left(\frac{1+n}{\lambda}E\|Y_{n+1}\|_{2}^{2}+E\|Y_{n+1}\|_{2}^{2}\right)\\
 & \le & \frac{C}{\sqrt{\lambda}}\frac{E\|Y_{n+1}\|_{2}^{2}}{1+n}\left(\frac{1+n}{\lambda}+1\right)\leq\frac{C}{\sqrt{\lambda}}D\left(\frac{1+n}{\lambda}+1\right)
\end{eqnarray*}
So when $\lambda/n^{2/3} \rightarrow \infty$ we have $\frac{1}{n+1}E_{D_{n+1}}\sum_{i=1}^{n+1}C_{3}^{[i]}  \rightarrow 0$.

This completes the proof of Theorem~\ref{thm:asymptotic}.
\end{proof}

\subsection{Proof of Lemma~\ref{lem:lips}}

\begin{proof}

We observe that the difference of the first derivatives for the function $\phi_\omega$ satisfies
\[
|\phi_{\omega}'(a)-\phi_{\omega}'(b)|=
\begin{cases}
2(1-\omega)|a-b| & \mathrm{if}\quad(a\leq0, b\leq 0),\\
2\omega|a-b| & \mathrm{if}\quad(a>0, b>0),\\
2|(1-\omega)a-\omega b| & \mathrm{if}\quad(a\leq0, b>0),\\
2|\omega a-(1-\omega) b| & \mathrm{if}\quad(a>0, b\leq 0).
\end{cases}
\]
Therefore we have
\begin{equation}\label{Lips1}
\vert \phi^{'}_{\omega}(a)-\phi^{'}_{\omega}(b)\vert \le L|a-b|  \quad \forall a,b,
\end{equation}
where $L=2\max(1-\omega,\omega)$. By the Lipschitz continuity of $\phi^{\prime}_{\omega}$ and Cauchy-Schwarz
inequality,
\begin{equation}
(\phi^{\prime}_{\omega}(a)-\phi^{\prime}_{\omega}(b))(a-b)\leq L|a-b|^{2}\qquad\forall a,b\in\mathbb{R}.\label{eq:monotonicity}
\end{equation}
If we let $\varphi_{\omega}(a)=(L/2)a^{2}-\phi_{\omega}(a)$, then \eqref{eq:monotonicity}
implies the monotonicity of the gradient $\varphi_{\omega}^{\prime}(a)=La-\phi^{\prime}_{\omega}(a)$.
Therefore $\varphi$ is a convex function and by the first order condition
for convexity of $\varphi_{\omega}$:
\[
\varphi_{\omega}(a)\geq\varphi_{\omega}(b)+\varphi_{\omega}^{\prime}(b)(a-b)\qquad\forall a,b\in\mathbb{R},
\]
which is equivalent to \eqref{eq:upper_bound}.
\end{proof}

\subsection{Proof of Lemma~\ref{lem:convergence}}

\begin{proof}
1. By the definition of the majorization function and the fact that
$\balpha^{(k+1)}$ is the minimizer in \eqref{eq:majobj}
\[
F_{\omega, \lambda}(\balpha^{(k+1)})\leq Q(\balpha^{(k+1)}\mid\balpha^{(k)})\leq Q(\balpha^{(k)}\mid\balpha^{(k)})=F_{\omega, \lambda}(\balpha^{(k)}).
\]

2. Based on \eqref{eq:alter_qbound} and the fact that $Q$ is continuous,
bounded below and strictly convex, we have
\begin{equation}
\mathbf{0}=\nabla Q(\balpha^{(k+1)}\mid\balpha^{(k)})=\nabla F_{\omega, \lambda}(\balpha^{(k)})+2\bK_{u}(\balpha^{(k+1)}-\balpha^{(k)}).\label{eq:trick}
\end{equation}
Hence
\begin{align*}
F_{\omega, \lambda}(\balpha^{(k+1)}) & \leq Q(\balpha^{(k+1)}\mid\balpha^{(k)})\\
 & =F_{\omega, \lambda}(\balpha^{(k)})+\nabla F_{\omega, \lambda}(\balpha^{(k)})(\balpha^{(k+1)}-\balpha^{(k)})+(\balpha^{(k+1)}-\balpha^{(k)})^{\intercal}\bK_{u}(\balpha^{(k+1)}-\balpha^{(k)})\\
 & =F_{\omega, \lambda}(\balpha^{(k)})-(\balpha^{(k+1)}-\balpha^{(k)})^{\intercal}\bK_{u}(\balpha^{(k+1)}-\balpha^{(k)}).
\end{align*}
By \eqref{eq:ku} and the assumption that $\sum_{i=1}^{n}\bK_{i}\bK_{i}^{\intercal}$
is positive definite, we see that $\mathbf{K}_{u}$ is also positive
definite. Let $\gamma_{\min}(\mathbf{K}_{u})$ be the smallest eigenvalue
of $\mathbf{K}_{u}$ then
\begin{equation}
0\leq\gamma_{\min}(\mathbf{K}_{u})\|\balpha^{(k+1)}-\balpha^{(k)}\|^{2}\leq(\balpha^{(k+1)}-\balpha^{(k)})^{\intercal}\bK_{u}(\balpha^{(k+1)}-\balpha^{(k)})\leq F_{\omega, \lambda}(\balpha^{(k)})-F_{\omega, \lambda}(\balpha^{(k+1)}).\label{eq:sandwitch}
\end{equation}
Since $F$ is bounded below and monotonically decreasing as shown
in Proof 1, $F_{\omega, \lambda}(\balpha^{(k)})-F_{\omega, \lambda}(\balpha^{(k+1)})$ converges to zero
as $k\rightarrow\infty$, from \eqref{eq:sandwitch} we see that $\lim_{k\rightarrow\infty}\|\balpha^{(k+1)}-\balpha^{(k)}\|=0$.

3. Now we show that the sequence $(\balpha^{(k)})$ converges to the
unique global minimum of \eqref{eq:obj}. As shown  in Proof 1,  the
sequence $(F_{\omega, \lambda}(\balpha^{(k)}))$ is monotonically decreasing, hence
is  bounded above. The fact that $(F_{\omega, \lambda}(\balpha^{(k)}))$ is bounded
implies that $(\balpha^{(k)})$ must also be bounded, that is because
$\lim_{\balpha\rightarrow\infty}F_{\omega, \lambda}(\balpha)=\infty$. We next show
that the limit of any convergent subsequence of $(\balpha^{(k)})$
is a stationary point of $F$. Let $(\balpha^{(k_{i})})$ be the subsequence
of $(\balpha^{(k)})$ and let $\lim_{i\rightarrow\infty}\balpha^{(k_{i})}=\widehat{\balpha}$,
then by \eqref{eq:trick}
\[
\mathbf{0}=\nabla Q(\balpha^{(k_{i}+1)}\mid\balpha^{(k_{i})})=\nabla F_{\omega, \lambda}(\balpha^{(k_{i})})+2\bK_{u}(\balpha^{(k_{i}+1)}-\balpha^{(k_{i})}).
\]
Taking limits on both sides, we prove that $\widehat{\balpha}$ is
a stationary point of $F$.
\begin{align*}
\mathbf{0} & =\lim_{i\rightarrow\infty}\nabla Q(\balpha^{(k_{i}+1)}\mid\balpha^{(k_{i})})=\nabla Q(\lim_{i\rightarrow\infty}\balpha^{(k_{i}+1)}\mid\lim_{i\rightarrow\infty}\balpha^{(k_{i})}).\\
 & =\nabla F_{\omega, \lambda}(\widehat{\balpha})+2\bK_{u}(\widehat{\balpha}-\widehat{\balpha})=\nabla F_{\omega, \lambda}(\widehat{\balpha}).
\end{align*}
Then by the strict convexity of $F$, we have that $\widehat{\balpha}$
is the unique global minimum of \eqref{eq:obj}.
\end{proof}

\subsection{Proof of Theorem~\ref{thm:iteration}}

\begin{proof}
1. By \eqref{eq:mathdef1} and \eqref{eq:majobj},
\begin{equation}
F_{\omega, \lambda}(\balpha^{(k+1)})\leq Q(\balpha^{(k+1)}\mid\balpha^{(k)})\leq Q(\Lambda_{k}\balpha^{(k)}+(1-\Lambda_{k})\widehat{\balpha}\mid\balpha^{(k)}).\label{eq:in1}
\end{equation}
Using \eqref{eq:LambdaK} we can show that
\begin{align}
 & Q(\Lambda_{k}\balpha^{(k)}+(1-\Lambda_{k})\widehat{\balpha}\mid\balpha^{(k)})\nonumber \\
= & F_{\omega, \lambda}(\balpha^{(k)})+(1-\Lambda_{k})\nabla F_{\omega, \lambda}(\balpha^{(k)})(\widehat{\balpha}-\balpha^{(k)})+(1-\Lambda_{k})^{2}(\widehat{\balpha}-\balpha^{(k)})^{\intercal}\bK_{u}(\widehat{\balpha}-\balpha^{(k)})\nonumber \\
= & \Lambda_{k}F_{\omega, \lambda}(\balpha^{(k)})+(1-\Lambda_{k})\left[Q(\widehat{\balpha}\mid\balpha^{(k)})-\Lambda_{k}(\widehat{\balpha}-\balpha^{(k)})^{\intercal}\bK_{u}(\widehat{\balpha}-\balpha^{(k)})\right]\nonumber \\
= & \Lambda_{k}F_{\omega, \lambda}(\balpha^{(k)})+(1-\Lambda_{k})F_{\omega, \lambda}(\widehat{\balpha}).\label{eq:in2}
\end{align}
Then the statement can be proved by substituting \eqref{eq:in2} into
\eqref{eq:in1}.

2. We obtain a lower bound for $F_{\omega, \lambda}(\widehat{\balpha})$
\begin{equation}
F_{\omega, \lambda}(\widehat{\balpha})\geq F_{\omega, \lambda}(\balpha^{(k)})+\nabla F_{\omega, \lambda}(\balpha^{(k)})(\widehat{\balpha}-\balpha^{(k)})+(\widehat{\balpha}-\balpha^{(k)})^{\intercal}\bK_{l}(\widehat{\balpha}-\balpha^{(k)}),\label{eq:lbound}
\end{equation}
and majorization $Q(\widehat{\balpha}\mid\balpha^{(k)})$
\begin{equation}
Q(\widehat{\balpha}\mid\balpha^{(k)})=F_{\omega, \lambda}(\balpha^{(k)})+\nabla F_{\omega, \lambda}(\balpha^{(k)})(\widehat{\balpha}-\balpha^{(k)})+(\widehat{\balpha}-\balpha^{(k)})^{\intercal}\bK_{u}(\widehat{\balpha}-\balpha^{(k)}).\label{eq:major}
\end{equation}
Subtract \eqref{eq:lbound} from \eqref{eq:major} and divide by $(\widehat{\balpha}-\balpha^{(k)})^{\intercal}\bK_{u}(\widehat{\balpha}-\balpha^{(k)})$,
we have
\begin{align}
\Lambda_{k} & =\frac{Q(\widehat{\balpha}\mid\balpha^{(k)})-F_{\omega, \lambda}(\widehat{\balpha})}{(\widehat{\balpha}-\balpha^{(k)})^{\intercal}\bK_{u}(\widehat{\balpha}-\balpha^{(k)})}\nonumber \\
 & \leq1-\frac{(\widehat{\balpha}-\balpha^{(k)})^{\intercal}\bK_{l}(\widehat{\balpha}-\balpha^{(k)})}{(\widehat{\balpha}-\balpha^{(k)})^{\intercal}\bK_{u}(\widehat{\balpha}-\balpha^{(k)})}\nonumber \\
 & \leq1-\gamma_{\min}(\bK_{u}^{-1}\bK_{l}).\label{eq:fineq}
\end{align}
Both $K_{u}$ and $K_{l}$ are positive definite by the assumption
that $\sum_{i=1}^{n}\bK_{i}\bK_{i}^{\intercal}$ is positive definite,
and since
\[
\bK_{u}^{-1}\bK_{l}=\bK_{u}^{-\frac{1}{2}}\bK_{u}^{-\frac{1}{2}}\bK_{l}\bK_{u}^{-\frac{1}{2}}\bK_{u}^{\frac{1}{2}},
\]
the matrix $\bK_{u}^{-1}\bK_{l}$ is similar to the matrix $\bK_{u}^{-\frac{1}{2}}\bK_{l}\bK_{u}^{-\frac{1}{2}}$,
which is positive definite. Hence
\[
\Gamma=1-\gamma_{\min}(\bK_{u}^{-1}\bK_{l})=1-\gamma_{\min}(\bK_{u}^{-\frac{1}{2}}\bK_{l}\bK_{u}^{-\frac{1}{2}})<1.
\]
By \eqref{eq:mathdef1} and \eqref{eq:fineq} we showed that $0\leq\Lambda_{k}\leq\Gamma<1$.

3. Since $\nabla F_{\omega, \lambda}(\widehat{\balpha})=\mathbf{0}$, using the Taylor
expansion on $F_{\omega, \lambda}(\balpha^{(k)})$ at $\widehat{\balpha}$, we have
\[
F_{\omega, \lambda}(\balpha^{(k)})-F_{\omega, \lambda}(\widehat{\balpha})\geq(\balpha^{(k)}-\widehat{\balpha})^{\intercal}\bK_{l}(\balpha^{(k)}-\widehat{\balpha})\geq\gamma_{\min}(\mathbf{K}_{l})\|\balpha^{(k)}-\widehat{\balpha}\|^{2},
\] 
\[
F_{\omega, \lambda}(\balpha^{(k)})-F_{\omega, \lambda}(\widehat{\balpha})\leq(\balpha^{(k)}-\widehat{\balpha})^{\intercal}\bK_{u}(\balpha^{(k)}-\widehat{\balpha})\leq\gamma_{\max}(\mathbf{K}_{u})\|\balpha^{(k)}-\widehat{\balpha}\|^{2}.
\]
Therefore, by Results 1 and 2 
\[
\|\balpha^{(k+1)}-\widehat{\balpha}\|^{2}\leq\frac{F_{\omega, \lambda}(\balpha^{(k+1)})-F_{\omega, \lambda}(\widehat{\balpha})}{\gamma_{\min}(\mathbf{K}_{l})}\leq\frac{\Gamma(F_{\omega, \lambda}(\balpha^{(k)})-F_{\omega, \lambda}(\widehat{\balpha}))}{\gamma_{\min}(\mathbf{K}_{l})}\leq\Gamma\frac{\gamma_{\max}(\mathbf{K}_{u})}{\gamma_{\min}(\mathbf{K}_{l})}\|\balpha^{(k)}-\widehat{\balpha}\|^{2}.
\] 
 
\end{proof}

\bibliographystyle{asa}
\bibliography{kere}

\end{document}